\documentclass[journal]{IEEEtran}  

\usepackage{cite}
\usepackage{url}
\usepackage{amsmath}
\usepackage{amssymb}
\usepackage{dsfont}
\usepackage{bm}
\usepackage{enumitem}
\usepackage[caption=false,font=footnotesize]{subfig}
\usepackage{siunitx}
\usepackage{algorithm}
\usepackage{graphicx}
\usepackage{epstopdf}
\usepackage{booktabs}
\usepackage{diagbox}
\usepackage{color}
\usepackage{soul}
\graphicspath{{figures/}}

\def\rainfty{\rightarrow\infty}

\def\diag{\text{diag}}

\def\sgn{\text{sgn}}

\newcommand{\mathletter}[1]{%
	\expandafter\newcommand\csname b#1\endcsname{\mathbb #1}
	\expandafter\newcommand\csname c#1\endcsname{\mathcal #1}
	\expandafter\newcommand\csname f#1\endcsname{\mathfrak #1}
	\expandafter\newcommand\csname til#1\endcsname{\widetilde #1}
	\expandafter\newcommand\csname ha#1\endcsname{\widehat #1}
	\expandafter\newcommand\csname bf#1\endcsname{\bf #1}
}%
\def\mathletters#1{\mathlettersB #1,,}
\def\mathlettersB#1,{\ifx,#1,\else\mathletter #1\expandafter\mathlettersB\fi}
\mathletters{A,B,C,D,E,F,G,H,I,J,K,L,M,N,O,P,Q,R,S,T,U,V,W,X,Y,Z}

\def\bee{\begin{equation}}
	\def\ene{\end{equation}}
\def\beq{\begin{eqnarray}}
	\def\enq{\end{eqnarray}}
\def\bmatri{\begin{bmatrix}}
\def\ematri{\end{bmatrix}}

\newtheorem{theo}{Theorem}

\newenvironment{proof}{\begin{IEEEproof}}{\end{IEEEproof}}

\title{Flight Control for UAV Loitering Over a Ground Target with Unknown Maneuver\\}

\author{Fei~Dong, Keyou~You,~\IEEEmembership{Senior Member,~IEEE}, Jiaqi~Zhang
	\thanks{*This work was supported  in part by the National Natural Science Foundation of China under Grant  61722308 and in part by the National Key  Research and Development Program of China under Grant 2017YFC0805310. ({\em Corresponding author: Keyou You})}
	 \thanks{The authors are with the Department of Automation, and Beijing National Research Center for Info. Sci. \& Tech. (BNRist), Tsinghua University, Beijing 100084, China. E-mail: \{dongf17, zjq16\}@mails.tsinghua.edu.cn, youky@tsinghua.edu.cn.}%
}

\begin{document}

\maketitle

\begin{abstract}
This paper proposes a flight controller for an unmanned aerial vehicle (UAV) to loiter over a ground moving target (GMT). We are concerned with the scenario that the stochastically time-varying maneuver of the GMT is {\em unknown} to the UAV, which renders it challenging to estimate the GMT's motion state.  Assuming that the state of the GMT is available, we first design a discrete-time Lyapunov vector field for the loitering guidance and then design a discrete-time integral sliding mode control (ISMC) to track the guidance commands. By modeling the maneuver process as a finite-state Markov chain, we propose a Rao-Blackwellised particle filter (RBPF),  which only requires a few number of particles, to simultaneously estimate the motion state and the maneuver of the GMT with a camera or radar sensor. Then, we apply the principle of certainty equivalence to the ISMC and obtain the flight controller for completing the loitering task. Finally, the effectiveness and advantages of our controller are validated via simulations.

\end{abstract}


\section{Introduction}
With the development of unmanned aerial vehicles (UAVs), using an UAV to track a ground moving target (GMT) has become an important trend  in both military and civilian applications, such as surveillance, border patrol, and convoy \cite{Zhang2010Vision,Ding2010Multi,Oliveira2016Moving,Meng2017Decentralized}.  To provide better aerial monitoring for ground intruders,  the UAV is required to loiter over the GMT with a desired distance.  In addition, a constant distance between an object and the camera sensor in the UAV can dramatically improve the quality of the vision data. The objective of this work is to design a flight controller for the UAV to loiter over a GMT with unknown maneuver. To achieve it, we need to address at least three challenging issues. 

The first is how to design guidance commands for the UAV to loiter over the GMT. Assuming that both the motion states of the GMT and the UAV are known, a geometrical approach has been exploited to design the guidance trajectory by analyzing the geometry relationship between the GMT and the UAV \cite{Dobrokhodov2008Vision}.  Obviously, this is of physical significance and is easy to understand. However, it is unable to provide many important kinematic variables, except the desired trajectory. For instance, it is unclear how to design the guidance command of the UAV's heading speed. To overcome this limitation, an approach using the continuous-time Lyapunov guidance vector field has been adopted in  \cite{Frew2007Lyapunov,Li2011Vision}. This approach guarantees that the guidance trajectory asymptotically converges to a circular orbit over the GMT with a desired radius at a certain speed. While the above mentioned approaches are for the continuous-time case, this work designs the {\em discrete-time} guidance commands by also using the Lyapunov vector field approach. This is generally more difficult than its continuous-time version. In fact, to ensure the effectiveness of the discrete-time commands,  the sampling frequency should be faster than an explicit lower bound, which is proportional to the maximum angular speed of the UAV as derived in this work. Clearly, there is no such an issue for the  continuous-time case. 

 The second is how to design a  robust controller for the UAV to track the discrete-time guidance commands in the presence of disturbances.  If the exact motion state of the GMT is available,  the proportional or proportional-derivative (PD) feedback laws are commonly used in  the continuous-time case, see e.g. \cite{Frew2007Lyapunov,Li2011Vision,Summers2009Coordinated,Zhang2010Vision,Pan2018Efficient,Kapitanyuk2017A,Loria2016Leader}. In \cite{Frew2007Lyapunov,Frew2008Coordinated}, the constant wind disturbances are considered, however, the wind velocities are known to the UAV. An adaptive estimator is designed to estimate the unknown constant wind velocities in \cite{Summers2009Coordinated}. A particular disturbance is studied in \cite{Xiao2017Target}, which is generated by a linear exogenous system with known structure and parameters. Accordingly, an estimator is proposed to handle this disturbance. Different from the PD control, the authors in \cite{Zhang2013Tracking} propose a tracking controller on the basis of the continuous-time sliding mode control (SMC) with a constant reaching law. It is worth mentioning that in \cite{Zhang2013Tracking}, a relative motion model between the GMT and the UAV is directly given by using their exact states. We refer the reader to \cite{Chwa2004Sliding,Shah2015Guidance} for an extended survey of the variable continuous-time SMC for path following problem. In this work, we design a discrete-time integral SMC (ISMC) \cite{Sarpturk1987On,Saaj2002A} via an integral sliding mode surface, and quantify how the sampling interval affects the tracking performance of the discrete-time  ISMC. Note that the designed guidance vectors can only be given {\em online}, i.e., guidance vectors after time $k$ are unavailable to the design of the $k$-th time input of the UAV. Advanced controllers do not always work, e.g., the model predictive control cannot be applied here as it relies on {\em future}  guidance vectors \cite{Camacho2004Model}.
	
 
The third is how to effectively estimate the motion state of the GMT with unknown maneuver, especially when the maneuver is stochastically time-varying. Note that the state of the UAV can usually be obtained by its position and orientation systems. The  estimation problem of the GMT state with known maneuver has been well studied by using a camera sensor or a radar sensor \cite{Zhang2010Vision,Dobrokhodov2008Vision,Ghommam2016Quadrotor}.  This can be easily solved via a nonlinear filter, e.g., the extended Kalman filter (EKF) \cite{Zhang2010Vision}. However, it is not sufficient to directly use an EKF to estimate the GMT state with unknown maneuver since this further introduces uncertainties to the dynamics of the GMT. To address it, we consider a stochastically time-varying maneuvering process  \cite{Li2004Survey}, and  model it as a finite-state Markov chain, whose state is introduced to represent a maneuver mode. For brevity and without loss of generality, we only consider three maneuver modes: keep straight, turn left and turn right of the GMT. Then, we design a Rao-Blackwellised particle filter (RBPF) \cite{zhang2018bayesian}  to simultaneously estimate the maneuver and the motion state of the GMT. The implementation and comparison between the standard PF and RBPF are well documented in \cite{doucet2001SMCintroduction,doucet2000rao,Gustafsson2002Particle}. 

The RBPF is provided to approximate the posterior distribution of the maneuver state, which is ternary valued and requires only a few number of particles. In all simulations,  we illustrate that 100 particles are sufficient  to achieve favorable estimation performance. Another advantage of the RBPF is that we do not need to access the exact transition probabilities between the maneuver modes. Once the  maneuver is known, we simply use the  EKF to estimate the motion state of the GMT. Thus, our filter exploits the advantages of both the RBPF and EKF.  This idea has been presented in the preliminary version of this work in \cite{dong2017}. Then, we adopt the principle of certainty equivalence \cite{bertsekas1995dynamic} and directly replace the true state of the GMT in the discrete-time guidance law and the discrete-time ISMC by its estimated version from the RBPF. 

Overall, our flight controller consists of three main components: (a) the discrete-time guidance commands for the desired loitering pattern; (b) the discrete-time ISMC; (c)  the RBPF to simultaneously estimate the maneuver modes and the state of the GMT.  The effectiveness of our controller is validated via simulation results.  

The rest of the paper is organized as follows. In Section II, the problem under consideration is formulated in details. Particularly, we explicitly describe the desired loitering pattern between the UAV and the GMT.  In Section \ref{secgui}, the guidance commands are designed by using the approach of the discrete-time  Lyapunov vector field.  In Section \ref{sec_ISMC}, we provide the discrete-time ISMC to track the guidance commands. In Section \ref{sec_motion}, we show how to design the RBPF with a camera sensor and a radar sensor, respectively, to estimate the motion state of the GMT. Simulations are conducted in Section \ref{sec_simulation}. And, some concluding remarks are drawn in Section \ref{sec_conclusion}.

\section{Problem Formulation}
We are concerned with the design of a flight controller for a fixed-wing UAV to loiter over a ground moving target (GMT), whose maneuver is modeled as a randomly unknown process. Specifically, relative to the GMT,  the UAV is expected to track a desired circle over the GMT with a constant angular speed. See Fig.~\ref{figpos} for illustrations where the desired orbit is to be tracked by the UAV with a desired radius at a constant speed. Due to the randomly unknown  maneuver of the GMT, this work is substantially different from \cite{Zhang2010Vision} where the maneuver of the GMT is essentially zero. Clearly, this is very restrictive as in many real applications, where we are required to track an uncooperative GMT.  A notable example is that the GMT is an intruder whose  maneuver is obviously time-varying and unknown to the UAV. To reduce the probability of being tracked, the intruder may further apply random maneuvers. 

In this section, we describe the dynamical models of the GMT and UAV as well as the sensor models in the UAV. 
     
\begin{figure}[!t]
	\centerline{\includegraphics[width=0.8\linewidth]{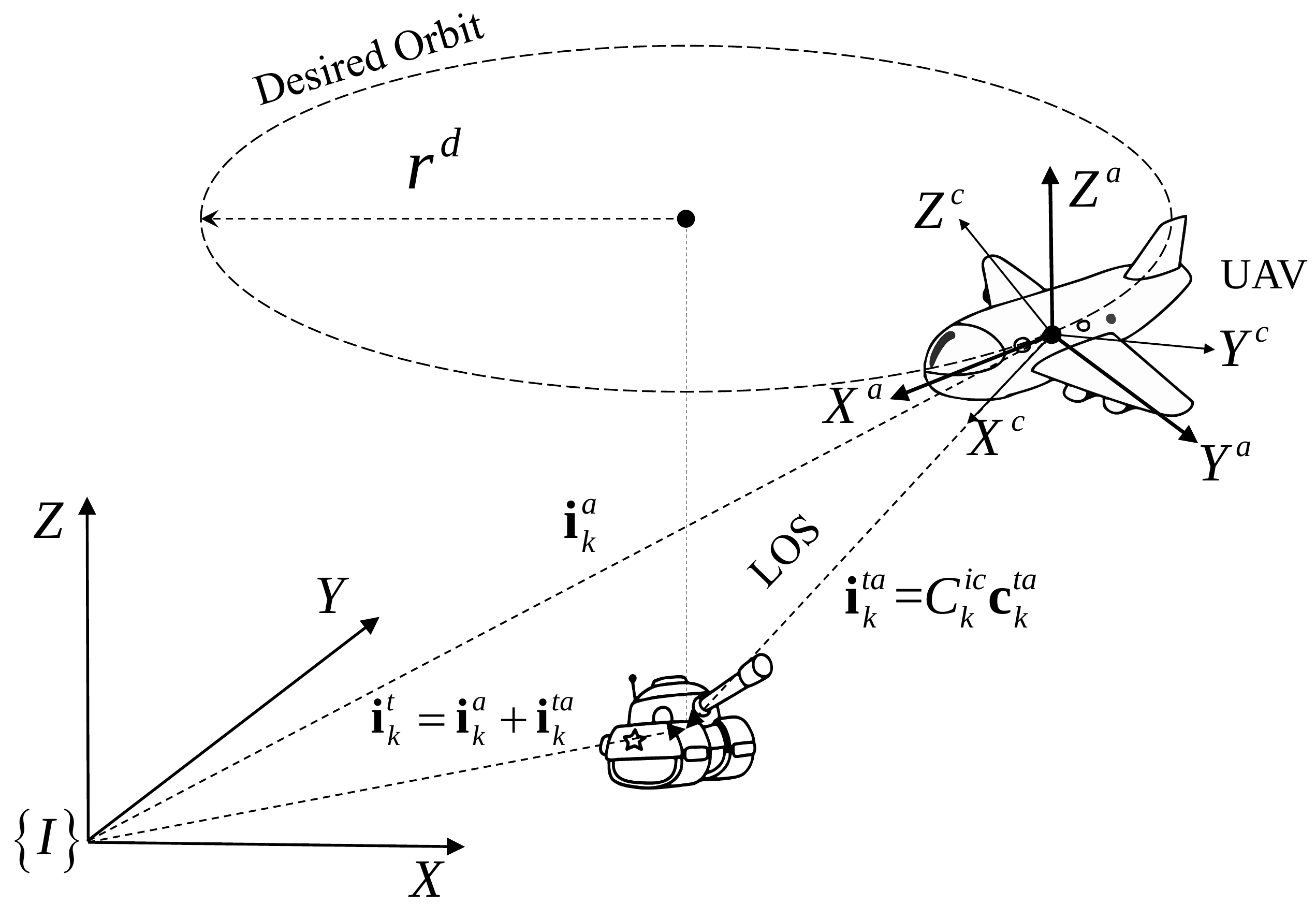}}
	\caption{Loitering over a ground target.}
	\label{figpos}
\end{figure}
	
	\subsection{Dynamical Model of the GMT  with Unknown Maneuver}
\label{subsec:GMT}
   The continuous-time version of the dynamical model of the GMT is given as
	\begin{equation}
		\begin{split}
			{\dot x^t} &= {v^{tx}}  \\
			{\dot y^t} &= {v^{ty}}  \\
			{\dot v^{tx}} &= u^{tx}+w^{tx}  \\
			{\dot v^{ty}} &= u^{ty}+w^{ty}
		\end{split}
	\end{equation}
	where $[{x^t},{y^t},{v^{tx}},{v^{ty}}]^T$ denotes the position and the velocity of the GMT on the horizontal plane respectively, and $[u^{tx},u^{ty}]^T$ is the maneuver of the GMT. Moreover, $[w^{tx},w^{ty}]^T$ represents the  random input noises, which are used to model the environmental disturbances.  
	
	Let $\tau$ be the sampling time interval. At time $k\tau$, we denote the state vector of the GMT as
	$$\bm{x}_{k}^t =  \bmatri{x_{k}^t}&{y_{k}^t}&{z_{k}^t}&{ v_{k}^{tx}}&{ v_{k}^{ty}}\ematri^T$$
		where $z^t_k$ denotes the coordinate of the GMT in $Z$-axis. Since the ground is usually not an ideal plane and may be subject to random fluctuations,  $z^t_k$ is time-varying as well. We model the  fluctuations as a Gaussian random process with zero mean. Then, the discrete-time dynamical model of the GMT is compactly given as
	\begin{equation}\label{gmtdyn}
		{\bm{x}}_{k + 1}^t = {{F}_k}\bm{x}_k^t + {{B}_k}\bm{u}^t({\gamma _k}) + G_k{\bm{w}_k^t}  
	\end{equation}
	where
	\begin{align*}
		&{F_k}{\rm{ = }}\bmatri
				1&0&0&{\tau}&0\\
				0&1&0&0&{\tau}\\
				0&0&1&0&0\\
				0&0&0&1&0\\
				0&0&0&0&1
			\ematri,~
			{B_k}{\rm{ = }}\bmatri
					{\tau^2/2}&0\\
					0&{\tau^2/2}\\
					0&0\\
					{\tau}&0\\
					0&{\tau}
				\ematri, ~\text{and}~ \\ 
				&~~~~~G_k{\rm{ = }}\bmatri
						{\tau^2/2}&0&0\\
						0&{\tau^2/2}&0\\
						0&0&{\tau}\\
						{\tau}&0&0\\
						0&{\tau}&0
					\ematri.
				\end{align*}
Moreover, $\{\bm{w}_k^t\} \subseteq {\mathbb{R}^{3}}$ are the white Gaussian input noises in X-axis and Y-axis, and the Gaussian fluctuations in Z-axis, i.e., ${\bm{w}_k^t}\sim {\mathcal N}(0,Q^t)$.
				
Different from \cite{Zhang2010Vision}, the maneuver of the GMT is generally non-zero and is unknown to the UAV.  In this work, the maneuver $\bm{u}^t({\gamma _k}) \in {\mathbb{R} ^{2}}$ is modeled as an unknown random vector  \cite{Li2004Survey}. Particularly, $\{\gamma _k\}$ is a three-state Markov chain, which corresponds to three modes of maneuver, e.g., keep straight, turn left, and turn right.  Note that our results can be easily generalized to the case with any finite number of states if computational resource is sufficient, and this number  directly imposes  constraints on the motion of the GMT. In the simulation, we also validate the case with nine states. 
				
Let $\mathcal{S} = \{ 1,2,3\}$ be the state space of the Markov chain $\{\gamma_k\}$, and denote its transition probability matrix by $P$, which actually can be estimated by the UAV with sensor measurements. For example, if we set 
				\begin{align} \label{eqtpmatrix}
					P  = \bmatri
							{0.9}&{0.05}&{0.05}\\
							{0.05}&{0.9}&{0.05}\\
							{0.05}&{0.05}&{0.9}
						\ematri,
					\end{align}
then the probability that the GMT continues to be in the state of moving forward  is 0.9 and  the probability that the GMT turns left or right from the state of moving forward is 0.05. If $P$ is unknown, we simply set each element of $P$ as $1/3$.   Without loss of generality, the GMT is initially set to be in the state of moving forward, i.e., $\gamma_0=1$.   
					
\subsection{Dynamical Model of the UAV}
We adopt a fixed-wing UAV to track the GMT, which is able to automatically keep its orientation stable. Compared to the flying altitude of the UAV, the fluctuations of the GMT in Z-axis is clearly small.  Thus, we are only interested in the scenario that the motion of the UAV  is restricted to a horizontal plane with a constant altitude. In \cite{Liu2015Discrete}, it provides a discrete-time dynamical model of a unicycle with a {\em constant} forward velocity. In this work, the forward velocity also needs to be controlled. Then, the discrete-time dynamical model of the UAV on the horizontal plane is described as
\begin{equation} \label{equ4}
\begin{split}
x_{k + 1}^a& = x_k^a + \frac{1}{ \tilde{u}_k^{a\psi } }\bigg (v_k^a(2\cos ( {\psi _k^a + \frac{ \tilde{u}_k^{a\psi} \tau}{2}} ){\rm{sin}}\frac{ \tilde{u}_k^{a\psi}\tau}{2}) \\
&~~~~+ \tilde{u}_k^{av}\tau \sin (\psi _k^a + \tilde{u}_k^{a\psi }\tau )\\
&~~~~+ \frac{ \tilde{u}_k^{av}}{\tilde{u}_k^{a\psi}}( - 2{\sin}(\psi _k^a + \frac{ \tilde{u}_k^{a\psi }\tau }{2})\sin\frac{ \tilde{u}_k^{a\psi }\tau}{2})\bigg ),  \\
y_{k + 1}^a& = y_k^a + \frac{1}{  \tilde{u}_k^{a\psi }}\bigg(2v_k^a (\sin(\psi _k^a + \frac{  \tilde{u}_k^{a\psi }\tau}{2})\sin\frac{ \tilde{u}_k^{a\psi }\tau}{2} ) \\
&~~~~- \tilde{u}_k^{av}\tau\cos (\psi _k^a + \tilde{u}_k^{a\psi }\tau) \\
&~~~~+ \frac{\tilde{u}_k^{av}}{\tilde{u}_k^{a\psi }}(2{\rm{cos}}(\psi _k^a + \frac{\tilde{u}_k^{a\psi }\tau}{2}){\rm{sin}}\frac{\tilde{u}_k^{a\psi }\tau}{2})\bigg),  \\
v_{{k} + 1}^a& =  v_{k}^a + {\tilde u}_k^{av}\tau,  \\
\psi _{k + 1}^a& = \psi _k^a + {\tilde u}_k^{a \psi } \tau,
\end{split}
\end{equation}
where $[x^a_k,{y^a_k},\psi_k^a,v^a_k]^T$  represents the coordinates, the heading direction, and the linear speed of the UAV on the plane, and   ${\tilde u}_k^{av} = {u}_k^{av} +{w}_k^{av}$ and ${\tilde u}_k^{a\psi} = {u}_k^{a\psi} +{w}_k^{a\psi}$ are the perturbed commanded acceleration (\si{m/{\s ^2}}) and turning rate (\si{rad/\s}) respectively. Here $\bm{w}_k^a := [w_k^{av},w_k^{a\psi}]^T \in \mathbb{R}^2$ represents the input disturbance, and is assumed to be bounded, i.e. 
\bee\label{bound}
|w_k^{av}| \le w^v,~ |w_k^{a\psi}| \le w^\psi.
\ene
The motion of the UAV is also restricted \cite{Ren2004Trajectory,Kang2009Linear,Beard2014Fixed}, e.g. \(\psi^a_k  \in [ - \pi ,\pi)\), and 
\begin{align}
u_k^{a\psi}\in\mathcal{U}^{a\psi} &:= \{ u^{a\psi} \in \mathbb{R}: |u^{a\psi}| \le u_{\max}^{a\psi}\} \label{maxanguspeed}
\end{align} 
where $u_{\max}^{a\psi}$ denotes the maximum angular speed of the UAV.

\subsection{Sensor Models in the UAV}
To complete the loitering task,  the UAV needs to be equipped with some necessary sensors. We first consider the camera sensor, and then the radar sensor, both of which are common in applications. 
 
\subsubsection{Camera Sensor}
A vision camera is amounted on a gimbal platform, which is to  adjust the camera to keep the line of sight (LOS) towards the target. By using the controller in \cite{Dong2016A}, the gimbal platform can maintain the GMT in the field of vision (FOV) of the camera to avoid the target loss. Thus, we directly assume that the GMT is always in the FOV of the vision camera, and is treated as a pure mass point. 

The camera projects a target point ${\bm p}:=[x,y,z]^T$ in 3D coordinates into a pixel point $[b,c]^T$ on the 2D camera image plane,  
see Fig.~\ref{figcam} for illustrations. 
\begin{figure}[t!]
	\centerline{\includegraphics[width=0.8\linewidth]{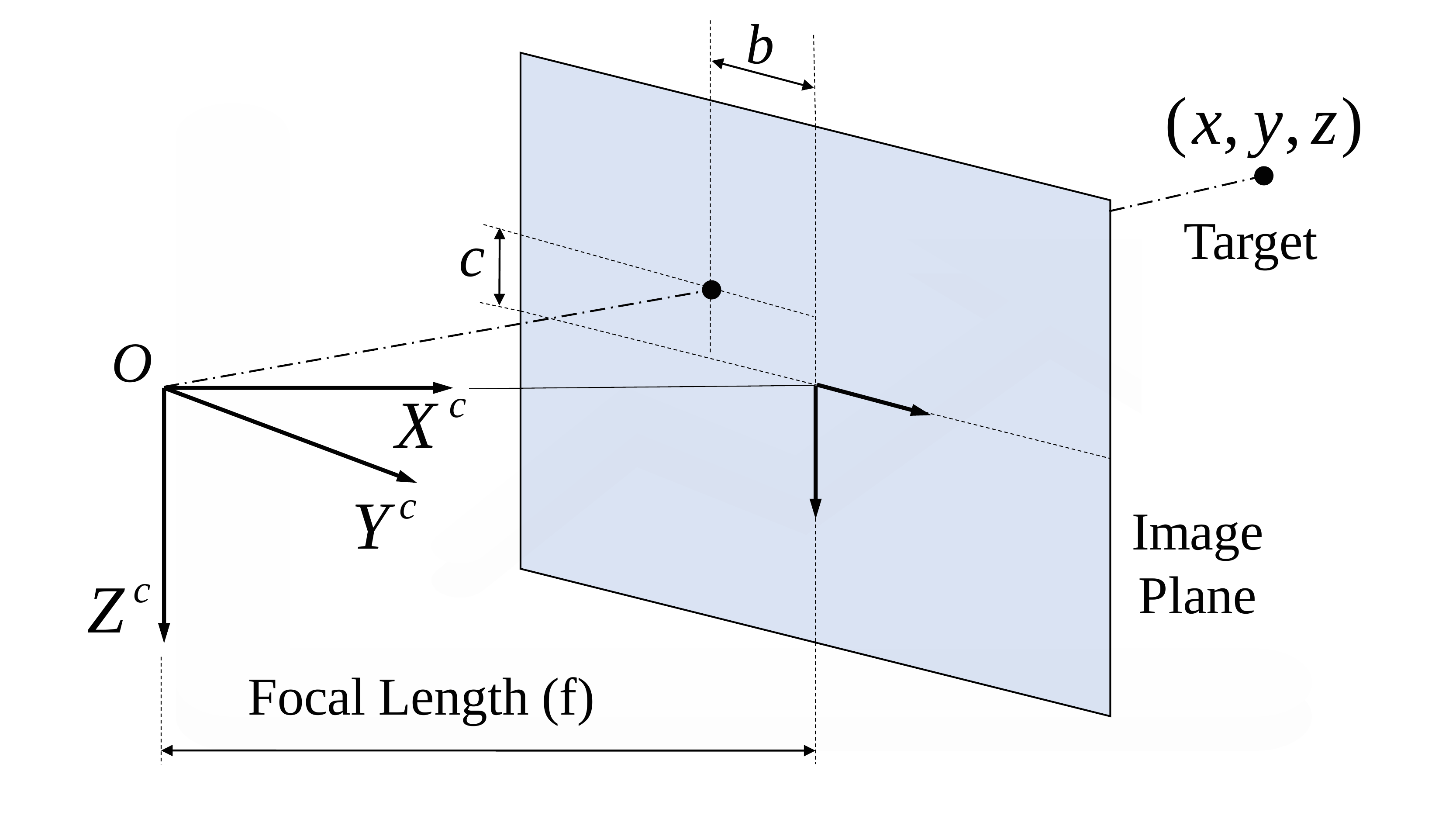}}
	\caption{Camera perspective projection model.}
	\label{figcam}
\end{figure}
The measurement model \cite{Zhang2010Vision} of the camera is thus expressed as 
\begin{equation} \label{eqcam}
	{\bm h}^c({\bm p}):=\begin{bmatrix}
			b\\
			c
			\end{bmatrix} = \frac{f}{x}\begin{bmatrix}y\\
		z
	\end{bmatrix}
\end{equation}
where $f$ is the camera focal length, and the image processing software can automatically produce the coordinate $[b,c]^T$ once the target is locked. To improve the image quality, we need to keep the object distance invariant.   

The coordinate $x=0$ corresponds to that the target point $[x, y, z]^T$ is exactly on the camera image plane, which is impossible here. Moreover, the Jacobian matrix of ${\bm h}^c({\bm p})$ is easily computed as
\begin{equation} 
J^c({\bm p})=\frac{f}{x^2}\begin{bmatrix}-y&x&0 \\ -z &0&x\end{bmatrix}. \label{jcbcam}
\end{equation}

\subsubsection{Radar Sensor}

A radar sensor is able to provide  the range and azimuth measurements between the radar and the target.  Its measurement model can be expressed as
\begin{equation} \label{eqrada}
{\bm h}^r({\bm p}):=\bmatri
	d\\
	\varphi 
	\ematri = \bmatri
	{ ({x^2} + {y^2} + {z^2})^{1/2} }\\
	{\arctan ({y}/{x})}
	\ematri,
\end{equation}
where $d$ and the  angle $\varphi$ are range and azimuth  measurements of the radar. Similarly, the Jacobian matrix of ${\bm h}^r({\bm p})$ is easily computed as
\begin{equation} 
J^r({\bm p})=\begin{bmatrix}x/d&y/d&z/d \\ -y/(x^2+y^2) &x/(x^2+y^2)&0\end{bmatrix}. \label{jcbradar}
\end{equation}

\subsection{The Objective of This Work}

The main objective of this paper is to design a flight controller for the UAV to loiter over the GMT by using a camera or radar sensor. 

To achieve it,  we design a tracking system as shown in Fig.~\ref{figsys}, which is mainly composed of a discrete-time guidance law, a discrete-time controller, and a motion estimator.  It is worth mentioning that the state $\bm{x}_{k}^a$ of the UAV  can be directly obtained by the navigation system, e.g., the position and orientation system (POS). The POS always consists of an inertial measurement unit (IMU) and a Global Navigation Satellite System (GNSS). The control of the gimbal platform has also been separately studied in \cite{Dong2016A}.  Thus, both the gimbal control and navigation are not the focus of this work.

\begin{figure}[!t]
	\centerline{\includegraphics[width=0.8\linewidth]{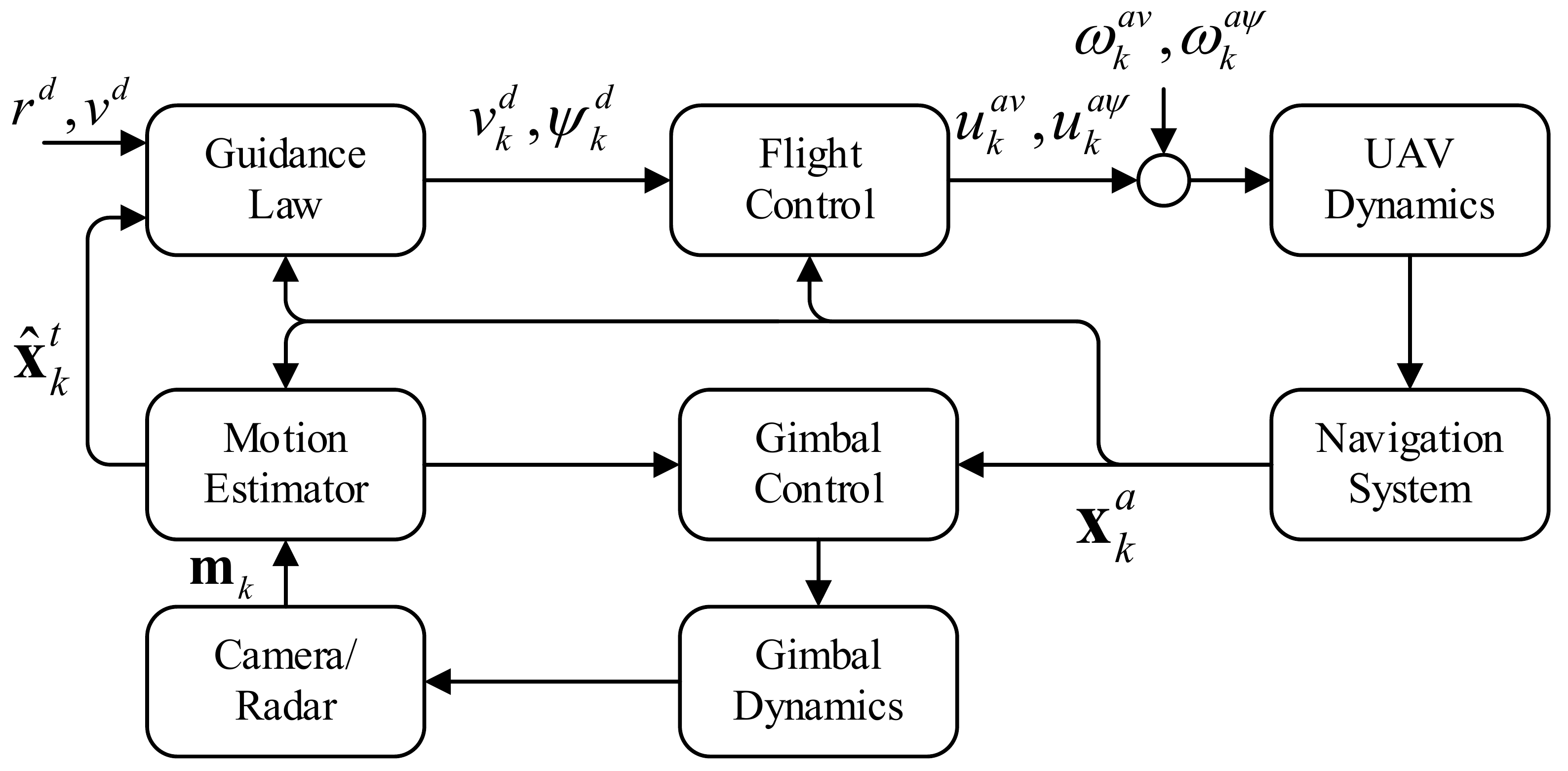}}
	\caption{Structure of the overall system.}
	\label{figsys}
\end{figure}

Overall,  the flight control problem of this work includes at least the following challenging issues:  (a) the discrete-time guidance law for the UAV, which extends the continuous-time guidance law in \cite{Zhang2010Vision} to the discrete-time case; (b) the control of the UAV to track the guidance trajectory. Since disturbances are unavoidable in the whole system,  the controller should be of sufficient robustness to uncertainties; (c) the motion state estimation of the GMT in the presence of the {\em unknown} maneuver, which requires to simultaneously estimate the maneuver and the motion state.

\section{Discrete-time Guidance Law} \label{secgui}
The discrete-time guidance law guides the UAV to loiter over the GMT with a desired radius at a relative constant speed. In \cite{Lawrence2003Lyapunov}, a continuous-time Lyapunov guidance vector field is proposed for a static target, which is extended to the case of a moving target with a constant velocity in  \cite{Frew2007Lyapunov} and \cite{Frew2008Coordinated}. Note that they all need the states of both the UAV and the GMT. Here, we design a discrete-time Lyapunov guidance vector field to direct the UAV to loiter over the GMT with a desired radius $r^d$ at a relative constant speed $v^d$.  

If the states of both the UAV and the GMT are available, let $x_k=x_k^a-x_k^t$ and $y_k=y_k^a-y_k^t$  represent the relative positions between the GMT and the UAV in X-axis and Y-axis, respectively. For a sampling period $\tau>0$, then 
$$v_k^x=\frac{x_{k+1}-x_k}{\tau}~\text{and}~v_k^y=\frac{y_{k+1}-y_k}{\tau}$$ are the relative velocities between the GMT and the UAV in X-axis and Y-axis.
\begin{figure}[!t]
	\centering		
	\includegraphics[width=0.8\linewidth]{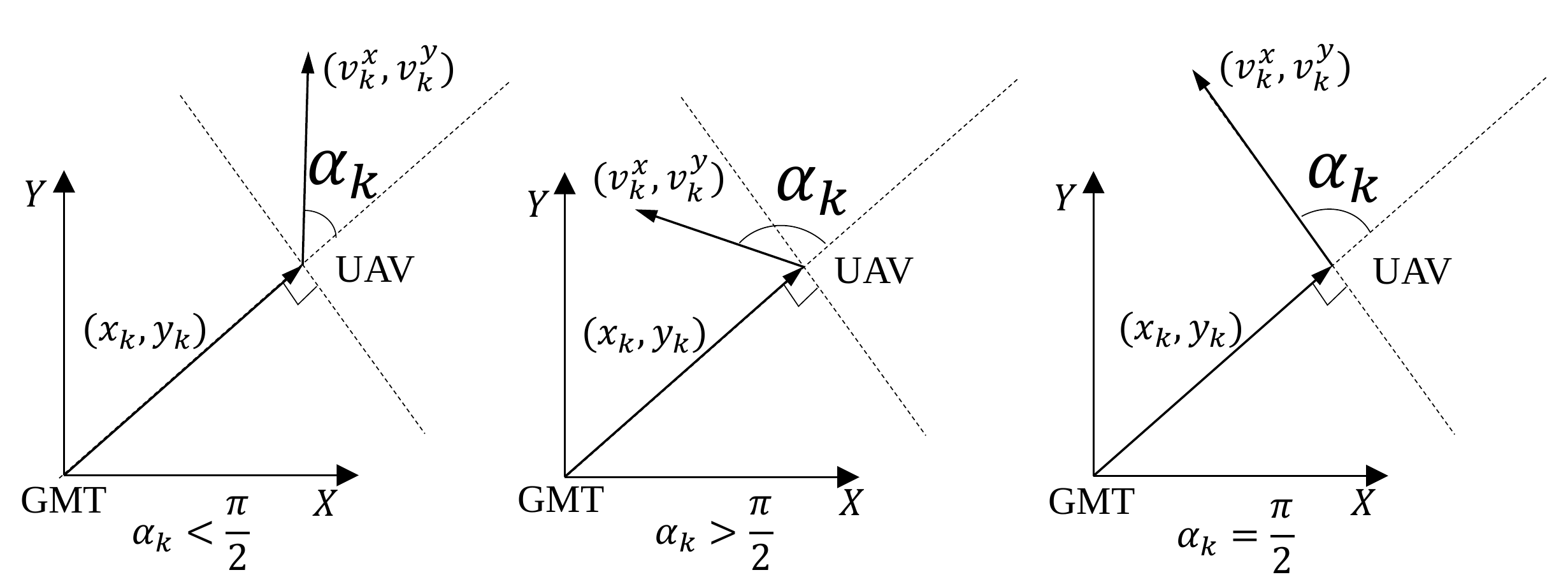}	
	\caption{Illustration of the discrete-time Lyapunov guidance vector method.}
	\label{figzuobiao}
\end{figure}

Given two positive constants $r^d$ and $v^d$, the discrete-time guidance vector is given by the following dynamical equation
\begin{equation} \label{eqlya}
\bmatri
	v_k^x\\
	v_k^y
	\ematri = \frac{{ - v^d}}{r_k\big((r_k)^2 + (r^d)^2\big)} 
\bmatri
	{{x_k}\big((r_k)^2 - (r^d)^2\big)+ {y_k}(2r^d r_k)}\\
	{{y_k}\big((r_k)^2 - (r^d)^2\big)- {x_k}(2r^d r_k)}
	\ematri 
\end{equation}
where  $r_{k} = ({(x_k)^2} + {(y_k)^2})^{1/2}$ denotes the projected relative distance between the GMT and the UAV onto the X-Y plane. 

Moreover, if the sampling interval $\tau$ is sufficiently small or the sampling frequency $1/\tau$ is fast enough, and the desired angular speed $v^d/r^d$ is not too large, i.e.,  
\bee
\frac{v^d}{r^d}<  u_{\max}^{a\psi}~\text{and}~\frac{1}{\tau}>\frac{\sqrt{3}u_{\max}^{a\psi}}{2}  \label{keyineq}
\ene
where $u_{\max}^{a\psi}$, defined in (\ref{maxanguspeed}), is the maximum angular speed of the UAV, then we show below that $r^d$  is  the desired radius of the loitering orbit, and $v^d$  is the desired loitering speed.

 To this end, we first provide an intuitive explanation on the discrete-time guidance vector in (\ref{eqlya}).  Let $ \alpha_k$  denote the angle between the relative position vector $[x_k, y_k]^T$ and relative velocity vector $[v^x_k, v^y_k]^T$. Then, its cosine can be computed as
\beq
\cos \alpha_k  &=& \frac{x_kv^{x}_ k+y_kv^{y}_ k}{\sqrt{(x_k)^2+(y_k)^2}{\sqrt{(v^x_ k)^2+(v^y_ k)^2}}} \nonumber\\
&=& - \frac{(r_k)^2 - (r^d)^2}{(r_k)^2 +( r^d)^2}\label{cosine}
\enq
where the first equality follows from the definition of the angle between two vectors, and the second equality is derived by using the design of $[v^{x}_ k,v^{y}_ k]^T$ in (\ref{eqlya}). 

If \(0<r _k< r^d\), it follows from (\ref{cosine}) that $\cos \alpha_k > 0$ and ${\pi }/{2}>\alpha_k >0$. Then, the projected relative distance \(r_k\) increases towards the desired radius. On the other side, if $r_k > r^d$,  then $\cos \alpha_k < 0$ and $\alpha_k > {\pi }/{2}$. This implies that \(r_k\) decreases towards the desired radius. If $r_k = r^d$, then $\cos \alpha_k = 0$, and $\alpha_k = {\pi }/{2}$, the relative velocity $[ v^x_k , v^y_k]^T$ is exactly orthogonal to the relative position vector $[x_k, y_k]^T$. Then, the UAV loiters over the GMT with the desired radius $r^d$  at the relative speed \(v^d\). See Fig. \ref{figzuobiao} for a graphical illustration. 

Now, we are in the position to provide a rigorous proof of the above observation. 
\begin{theo}  If the relative position $[x_k,y_k]^T$ and velocity $[v_k^x,v_k^y]^T$ on the X-Y plane are designed via the discrete-time guidance law (\ref{eqlya}) where $r^d$ and $v^d$ satisfy (\ref{keyineq}),  the UAV eventually loiters over the GMT with a desired radius $r^d$ at an angular speed $v^d/r^d$. 
\end{theo}

\begin{proof}
Let the coordinates $x_k$ and $y_k$ be converted to the polar coordinates $r_k$ and $\theta_k$ by using the trigonometric functions, i.e., 
\begin{align*}
x_k &= r_k\cos \theta_k,  \\
y_k &= r_k\sin \theta_k . 
\end{align*}
Then, the discrete-time vector field of \eqref{eqlya} in the polar coordinates is given as
\begin{equation} \label{eq19}
\bmatri
	{\Delta r_k}\\
	{{r_k}\Delta {\theta _k}}
	\ematri = -v^d {\tau}\bmatri
	\displaystyle{\frac{(r_k)^2 - (r^d)^2}{(r_k)^2 + (r^d)^2}}\\
	\displaystyle{\frac{-2r^d r_k}{(r_k)^2 + (r^d)^2}}
	\ematri
\end{equation}
where $\Delta r_k=r_{k+1}-r_k$ and $\Delta {\theta _k}=\theta_{k+1}-\theta_k$.

Consider the following Lyapunov function candidate
$$
{V(r_k)} = \frac{1}{2}  \left((r_k)^2 - (r^d)^2 \right)^2.
$$
Then, taking the difference of $V(r_k)$ along (\ref{eqlya}) leads to that
\begin{align}
\Delta {V(r_k)}  &= {V(r_{k+1})} - {V(r_k)}\nonumber \\
&= \frac{1}{2}\left((r_{k+1})^2 + (r_k)^2 - 2(r^d)^2\right)(r_{k+1} + r_k )(\Delta r_k) \nonumber \\
 &= \frac{1}{2}(r_{k + 1} + r_k)(\Delta r_k)^2 \big(2\frac{(r_k)^2 + (r^d)^2}{ - {v^d} \tau} +\nonumber  \\
 & ~~~~\Delta r_k + 2r_k\big)  \nonumber \\
 &=  - \frac{1}{2}(r_{k + 1} + r_k)(\Delta r_k)^2\big(2\frac{(r_k)^2 + (r^d)^2 }{{v^d} \tau } - 2 r_k + \nonumber \\
&~~~~v^d \tau \frac{(r_k)^2 - (r^d)^2}{(r_k)^2 + (r^d)^2}\big).\label{diffly}
\end{align}

Since $-1\le \displaystyle \frac{(r_k)^2 - (r^d)^2}{(r_k)^2 + (r^d)^2} <  1$, then
$v^d \tau \displaystyle \frac{(r_k)^2 - (r^d)^2}{(r_k)^2 + (r^d)^2} \ge  - v^d \tau,$
 which together with \eqref{diffly} implies that
\begin{align*}
\Delta {V(r_k)} &\le  
 - \frac{1}{2}(r_{k + 1} + r_k)(\Delta r_k^2)\frac{1}{v^d \tau} \big( 2(r_k)^2 -  2 {v^d}\tau r_k   \\
&~~~~  + 2(r^d)^2 -{({v^d}\tau )^2}  \big).
\end{align*}

Thus, the sign of $\Delta {V(r_k)}$ is determined by the following quadratic term
\begin{align}
2(r_k)^2 -  2 {v^d}\tau r_k  + 2(r^d)^2 -{({v^d}\tau )^2}.\label{quad}
\end{align}

By (\ref{keyineq}), one can easily verify that $
(-  2 {v^d}\tau )^2 - 4 \times 2\times\left (2(r^d)^2 -({v^d}\tau)^2\right) <0$, 
which implies that (\ref{quad}) is always positive for any $r_k$. 

By combining the above, it finally holds that $$\Delta V(r_k) \le 0.$$ 

Moreover,  $V(r_k)$ has the following three properties
\begin{itemize}
	\item (nonnegative) $V(r^d) = 0$ and $V(r_k) > 0$, $\forall r_k \not= r^d$,
	\item (strictly decreasing) $\Delta V(r_k) < 0$, $\forall r_k \not= r^d$,
	\item (radially unbounded) if $r_k\rightarrow \infty$, then $V(r_k) \rightarrow \infty$.
\end{itemize}

By the discrete-time version of Theorem 4.2 in \cite{Khalil2002Nonlinear}, then $r_k =r^d$ is an equilibrium point, which is also globally asymptotically stable. That is, $\lim_{k\rainfty}r_k=r^d$. This implies that the relative distance $r_k$ on the X-Y plane between the UAV and the GMT eventually converges to the desired radius $r^d$.

In view of (\ref{eq19}), it follows that $(\Delta r_k)^2 + (r_k \Delta \theta_k)^2 = (v^d \tau)^2$.  When the UAV is flying on the circular orbit, i.e. $\Delta r_k =0$, it only has the  tangential speed $(r_k \Delta \theta_k) / \tau$, which is equal to $v^d$. 
\end{proof}

\begin{figure}[t!]
	\centerline{\includegraphics[width=0.6\linewidth]{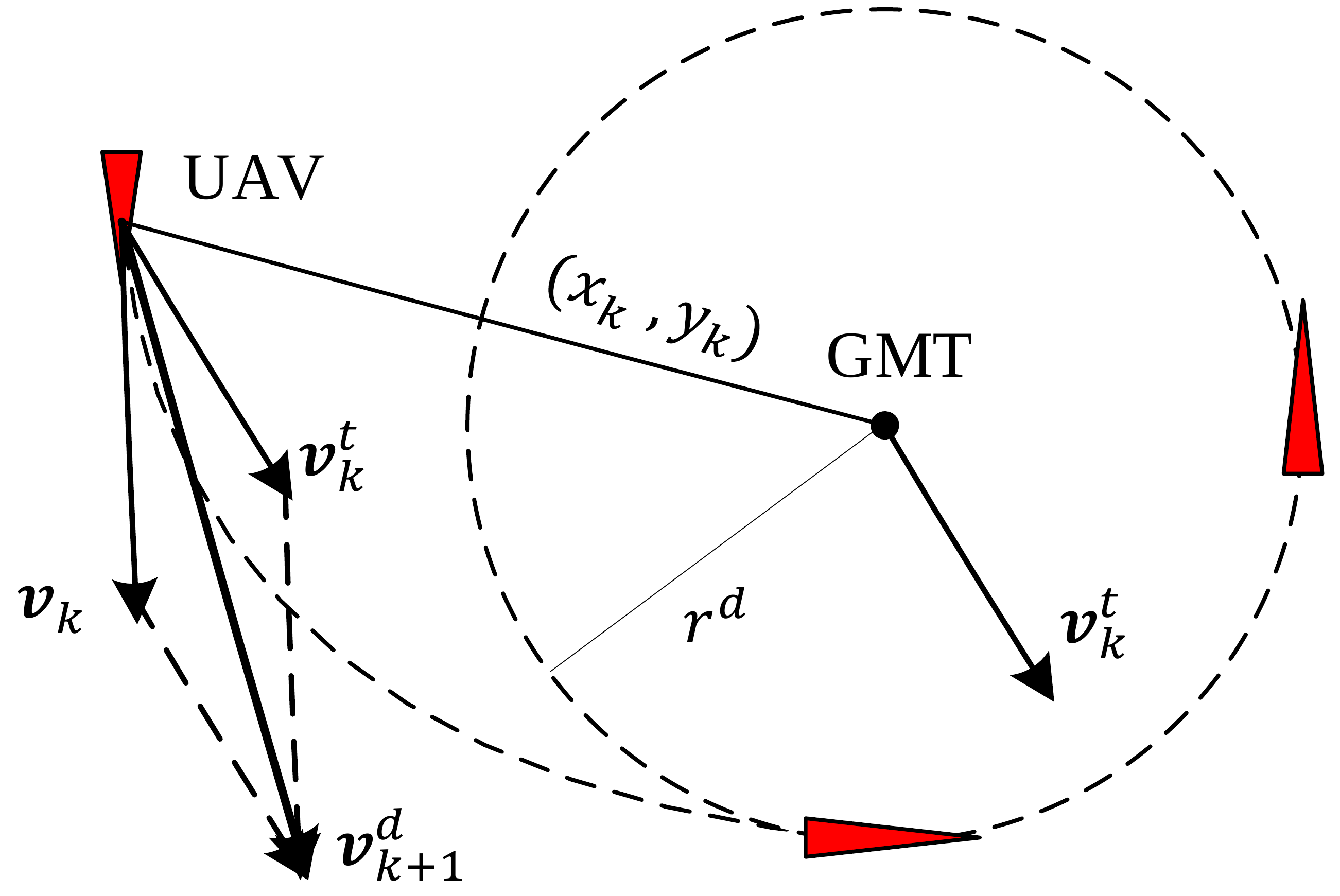}}
	\caption{The discrete-time guidance vector field.}
	\label{figlya}
\end{figure}

If the state of the GMT is known, the desired speed and heading angle of the UAV  in the inertial frame is given as
\begin{equation} \label{equmov}
\begin{split}
v^d_{k+1} &= \sqrt {(v_k^{tx}+v_k^x)^2  + ( v_k^{ty}+v_k^y )^2 }, \\
{\psi ^d_{k+1}} &= \arctan \frac{v_k^{ty}+v_k^y }{v_k^{tx}+v_k^x},\end{split}
\end{equation}
where $\bm{v}_k^t =[ v_k^{tx},v_k^{ty}]^T $ is the velocity of the GMT, and the relative velocity $\bm{v}_k=[ v_k^{x},v_k^{y}]^T$ on the X-Y plane is designed in (\ref{eqlya}). See Fig.~\ref{figlya} for illustrations. 

\section{Discrete-time Integral SMC with Perfect State of the GMT}
\label{sec_ISMC}

In this section, we design the controller for the UAV to asymptotically track the guidance commands $\{v^d_k\}$ and $\{\psi^d_k\}$ in (\ref{equmov}).  To handle the disturbance $\bm{w}_k^a$, we design a  discrete-time integral sliding model controller (ISMC), which is different from the continuous-time proportional-derivative (PD) controller in \cite{Zhang2010Vision,Frew2008Coordinated}. Usually, the ISMC has strong robustness to system uncertainties and external disturbances.

Define the  tracking errors of the desired speed and heading angle by
	\[{\bm{e}_k} := \bmatri
		{e_k^{v}}&{e_k^{\psi}}
		\ematri^T= \bmatri
		{v^a_k - v_k^d}\\
		{\psi^a_k  - \psi _k^d}
		\ematri
		\]
where $[v_k^d,\psi_k^d]^T$ is given in (\ref{equmov}). Together with the UAV dynamics in (\ref{equ4}), the dynamical equation of the tracking errors can be written as follows
 \beq 
	{\bm{e}_{k+1}} = \bmatri
	{e_{k+1}^{v}}\\
	{e_{k+1}^{\psi}}
	\ematri= \bmatri
	{e^v_k + u_k^{av}\tau +w_k^{av}\tau - \Delta v_k^d}\\
	{e^\psi_k + u_k^{a\psi}\tau +w_k^{a\psi}\tau - \Delta \psi_k^d}
	\ematri
\enq
where $\Delta v_k^d = v_{k+1}^{d} - v_{k }^{d}$ and $\Delta \psi_k^d = \psi_{k+1}^{d}- \psi_{k}^{d}$. Then, the  flight controller of the UAV is designed by using the ISMC, which is explicitly given below
\begin{equation} \label{equcon}
\bmatri
	{u_k^{av}}\\
	{u_k^{a\psi}}
	\ematri =
- W\sgn ({ \bm{s}_{k}}) - {M} \bm{s}_{k} - {C}\bm{e}_k  + \frac{1}{\tau}\bmatri
	{\Delta v_k^d }\\
	{\Delta \psi _k^d }
	\ematri 
\end{equation}
where the diagonal matrices ${W } = \diag({{w ^v}},{{w ^\psi}})$, ${M } = \diag({{m ^v}},{{m ^\psi}})$, and $C = \diag({c^v} ,{c^\psi})$  are to be designed, and the sign function ${\rm sgn} ({s})$ is defined as 
\begin{align*}
{\rm sgn} (s) = \left\{ {\begin{array}{*{20}{c}}
	{~~1,}&{s > 0,}\\
	{~~0,}&{s = 0,}\\
	{ - 1,}&{s < 0.}
	\end{array}}
	\right. 
\end{align*}

Moreover, the discrete-time integral sliding mode surface is designed as
\bee
\label{slidesurface}
{\bm {s}_k} = [{s_k^v}~~{s_k^\psi}]^T= \bm{e}_k +\tau {C}\sum\limits_{i = 0}^{k-1} {\bm{e}_i}. \ene

Before proving the effectiveness of the ISMC in (\ref{equcon}), we introduce some notations in this section. For a vector ${\bm x}$, then $|{\bm x}|$ takes the absolute values over every element of ${\bm x}$, $\sgn({\bm x})$ applies the sign function to each element of ${\bm x}$, and  $\diag({\bm x})$ is a diagonal matrix whose diagonal element  exactly corresponds to each element of ${\bm x}$. 
For two vectors ${\bm x}$ and ${\bm y}$, the relation ${\bm x}\succ {\bm y}$ (${\bm x}\prec {\bm y}$) means that each element of ${\bm x}$ is strictly greater (less) than the  element of the same position in ${\bm y}$.

\begin{theo} \label{thm_error}Consider the discrete-time dynamical model of the UAV in (\ref{equ4}).  Under the ISMC in (\ref{equcon}) and select the design parameters as follows:
$0<c^v, c^\psi<1/\tau$, $0< m^v, m^\psi<1/\tau$.  Then,  the UAV asymptotically tracks the guidance commands $[v_k^d,\psi_k^d]^T$ in (\ref{equmov}) with the tracking error $\bm{e}_k$ satisfying that
\begin{align*} 
\limsup_{k\rainfty}|\bm{e}_{k}|\prec 4 C^{-1}  (2 I-\tau M)^{-1}\bmatri {w^v }\\ {w^\psi} \ematri,
\end{align*}
where $w^v$ and $w^\psi$ are given in (\ref{bound}).
Moreover, if $2C=2M=\frac{1}{\tau} I$, then $$\limsup_{k\rainfty}|\bm{e}_{k}|\prec \frac{16}{3}\tau \bmatri {w^v }\\ {w^\psi} \ematri.$$
\end{theo}

\begin{proof} Clearly, the discrete-time exponential reaching law is expressed as
\begin{align} \label{eq31}
\Delta{\bm {s}_{k }}= {\bm {s}_{k+1}} - {\bm {s}_{k }} =  -\tau W \sgn ({\bm {s}_k}) - \tau M {\bm {s}_k} .
\end{align}
Together with the ISMC in (\ref{equcon}), we can easily obtain that 
\begin{align} \label{eq32}
{\bm {s}_{k+1}}  =  (I - \tau M) {\bm {s}_k}  -\tau W \sgn ({\bm {s}_k}) + \tau \bm{w}_k^a.
\end{align}

Define a Lyapunov function candidate as $$V^s(\bm{s}_k) = \frac{1}{\tau}\bm{s}_k ^T \bm{s}_k.$$ 
Taking the difference of $V^s(\bm{s}_k)$ along (\ref{eq32}), we can obtain that 
 \begin{align}
 &\hspace{-0.5cm}\Delta V^s(\bm{s}_k) = V^s(\bm{s}_{k+1})-V^s(\bm{s}_k)\label{deltav}\\
 & =  (\bm{s}_{k+1} + \bm{s}_{k })^T
 \big(- M {\bm {s}_k}  - W \sgn ({\bm {s}_k})+\bm{w}_k^a\big)\nonumber \\
 &\le - (\bm{s}_{k+1} + \bm{s}_{k })^T
 (  M {\bm {s}_k})  - (\bm{s}_{k+1} + \bm{s}_{k })^T  W \sgn ({\bm {s}_k})\nonumber \\ 
 &~~~~+|\bm{s}_{k+1} + \bm{s}_{k }|^T( \bmatri w^v, w^\psi \ematri^T)\nonumber\\
 &= - (\bm{s}_{k+1} + \bm{s}_{k })^T( M {\bm {s}_k})\nonumber \\
 &~~~~- (\bm{s}_{k+1} + \bm{s}_{k })^T  W (\sgn ({\bm {s}_k})-\sgn ({\bm {s}_k}+{\bm {s}_{k+1}})).\nonumber
  \end{align}
  
 If $\diag(\bm{s}_{k+1} + \bm{s}_{k })  \sgn ({\bm {s}_k})\succ 0$, it holds that $(\bm{s}_{k+1} + \bm{s}_{k })^T \sgn ({\bm {s}_k})>0$. 
Together with (\ref{eq32}), we have 
  \begin{align*}
     &\diag(\bm{s}_{k+1} + \bm{s}_{k })  \sgn ({\bm {s}_k}) \\
      & =\diag \big ( (2I -\tau M) \bm s_k -\tau W \sgn(\bm s_k)+ \tau \bm w_k^a  \big) \sgn ({\bm {s}_k}) \\
      &\ge \diag\big( (2I -\tau M) \bm s_k \big) \sgn(\bm s_k) -2\tau( \diag(w^v,w^\psi)  ) |\sgn(\bm s_k)| \\
      &=  (2I -\tau M)  |\bm s_k| -2\tau  \bmatri w^v, w^\psi \ematri^T    
  \end{align*} 
 Thus, if $|{\bm s}_k|  \succ (2I-\tau M)^{-1}(2\tau \bmatri w^v, w^\psi \ematri^T)$, the reaching condition holds, i.e.,
 \beq
 \diag(\bm{s}_{k+1}+\bm{s}_k) \sgn(\bm{s}_k) \succ 0,    
 \enq
 which implies that $\sgn(\bm{s}_k)=\sgn(\bm{s}_{k+1}+\bm{s}_k)$. Together with (\ref{deltav}), it follows that $\Delta V^s(\bm{s}_k)<0$. That is, the boundary layer is attractive. By \cite{Du2016Chattering}, there exists a finite time $k_0$ such that  
 \beq \label{eq35}
|{\bm s}_k|  \prec (2 I-\tau M)^{-1}(2\tau \bmatri w^v, w^\psi \ematri^T), \forall k \ge {k_0}.
 \enq

Next, we show that the tracking error ${\bm e}_k$ will also be attracted to a bounded region.  To elaborate it, let $\Delta {\bm e}_k:= {\bm e}_{k+1}- {\bm e}_k$. It follows from (\ref{slidesurface}) that $$\bm{e}_k=\bm{s}_k -\tau{C}\sum\limits_{i = 0}^{k-1} {\bm{e}_i}$$
and
$\Delta \bm{e}_k=\bm{s}_{k+1}-\bm{s}_k-\tau C \bm{e}_k$. 

For any $k\ge k_0$, the above implies that
\begin{align}\label{errordyn}
\bm{e}_{k+1}&=(I-\tau C)\bm{e}_k+\Delta \bm{s}_k \nonumber \\
&=(I-\tau C)^{k+1}\bm{e}_0+\sum_{i=0}^{k}(I-\tau C)^{k-i} \Delta\bm{s}_i \nonumber \\
&=(I-\tau C)^{k+1}\bm{e}_0 +\sum_{i=0}^{k_0-1}(I-\tau C)^{k-i} \Delta\bm{s}_i \nonumber\\
&~~~~~ +  \sum_{i=k_0}^{k}(I-\tau C)^{k-i} \Delta\bm{s}_i.
\end{align}

In light of  (\ref{eq31}) and (\ref{eq35}), it is clear that 
\begin{align*}  
|\Delta{ \bm s}_k|  \prec 4\tau (2I-\tau M)^{-1}\bmatri {w^v }, {w^\psi} \ematri^T, \forall k \ge k_0.
\end{align*}
Thus, $ \Delta\bm{s}_i$ is uniformly bounded. Since the spectral radius of $I-\tau C$ is strictly less than one, it implies that  $\lim_{k\rainfty}(I-\tau C)^{k+1}\bm{e}_0=0$, $\lim_{k\rainfty}\sum_{i=0}^{k_0-1}(I-\tau C)^{k-i} \Delta\bm{s}_i=0$ and 
\begin{align*}  
&\left|\sum_{i=k_0}^{k}(I-\tau C)^{k-i} \Delta\bm{s}_i\right| \\
&\prec 4\tau \sum_{i=k_0}^{k}(I-\tau C)^{k-i}  (2I-\tau M)^{-1}\bmatri {w^v }, {w^\psi} \ematri^T \\
&\prec 4\tau \sum_{i=0}^{\infty}(I-\tau C)^{i}  (2I-\tau M)^{-1}\bmatri {w^v }, {w^\psi} \ematri^T \\
& \prec 4 C^{-1}  (2I-\tau M)^{-1}\bmatri {w^v }, {w^\psi} \ematri^T. 
\end{align*}

Together with (\ref{errordyn}), it follows that
\begin{align*} 
\limsup_{k\rainfty}|\bm{e}_{k+1}|\prec 4 C^{-1}  (2I-\tau M)^{-1}\bmatri {w^v }, {w^\psi} \ematri^T.
\end{align*}
The rest of the proof is trivial. 
\end{proof}
By Theorem \ref{thm_error}, the tracking error is proportional to the sampling period $\tau$ and the size of disturbances to the UAV, which clearly is consistent with our intuition.   

\section{Motion Estimation of the GMT with Unknown Maneuver}\label{sec_motion}
If the state of the GMT is perfectly known, we have designed the discrete-time guidance law and the ISMC for the UAV  in the previous sections. Since the GMT can be an intruder or enemy, it is impossible for the UAV to access its {\em exact} maneuver and thus the state cannot be accurately obtained. 
From the tracking system in Fig.~\ref{figsys}, it is clear that the motion  estimation is vital for the guidance law, the flight control and the gimbal control for stabilizing the camera sensor.  If the maneuver is known, the state estimation problem of the GMT is well studied by using the standard nonlinear filters, e.g., EKF, Unscented Kalman filter (UKF) or PF \cite{Arulampalam2002A}.   To address this case with unknown maneuver, we consider using a Markov chain to model the maneuver process in Section \ref{subsec:GMT}, and adopting our recently proposed Rao-Blackwellised particle filter (RBPF) \cite{zhang2018bayesian} to simultaneously estimate both the maneuver modes $\gamma_k$ and the state  of the GMT. Compared to the standard PF, the number of sampling particles for the RBPF is much smaller.    

\subsection{State Estimation of the GMT}
In virtue of (\ref{gmtdyn}), the discrete-time dynamical model of the GMT and the measurement equation are collectively given as 
\begin{equation}
\begin{split}
{\bm{x}}_{k + 1}^t &= {{F}_k}\bm{x}_k^t + {{B}_k}\bm{u}^t({\gamma _k}) + G_k{\bm{w}_k} \\
\bm{m}_k &= {\bm h}(\bm{x}_k^t) + \bm {v}_k
\end{split}
\end{equation}
where $\bm{m}_k\in\bR^2$ is the  measurement of the UAV by using a camera or a radar sensor and $\{\bm {v}_k\}$ is the measurement Gaussian white noise, i.e., ${\bm {v}_k}\sim {\mathcal N}(0,R)$. 

Define $\mathit{\Gamma} _k = \{ \gamma _0,...,\gamma _k\} $, $\mathcal{M}_k = \{ \bm{m}_0,...,\bm{m}_k\} $.   Recall that the minimum variance estimate of  $\bm{x}_k^t$ is given as
\begin{equation} \label{eq25}
\begin{split}
\bm{\widehat {x}}_{k|k}^{t} &= \int {\bm{x}_k^t} p(\bm{x}_k^t|\mathcal{M}_k)\mathrm{d}\bm{x}_k^t \\
&= \iint {\bm{x}_k^t} p(\bm{x}_k^t,\mathit{\Gamma} _{k - 1}|\mathcal{M}_k)\mathrm{d}\bm{x}_k^t \mathrm{d}\mathit{\Gamma} _{k - 1}.
\end{split}
\end{equation}

Notice that  $p(\bm{x}_k^t,\mathit{\Gamma} _{k - 1}|\mathcal{M}_k)$ is not Gaussian and is impossible to be analytically obtained. Thus, the integral is not computable and we have to resort to a numerical approach.

To exposit it, it follows from the law of total probability that 
\[p(\bm{x}_k^t,\mathit{\Gamma} _{k - 1}|\mathcal{M}_k) = p(\bm{x}_k^t|\mathit{\Gamma} _{k - 1},\mathcal{M}_k)p(\mathit{\Gamma} _{k - 1}|\mathcal{M}_k)\]
where $p(\bm{x}_k^t|\mathit{\Gamma} _{k - 1},\mathcal{M}_k)$ is approximately computed by the extended Kalman filter (EKF) in a recursive form \cite{anderson1979optimal}.

However, $p(\mathit{\Gamma} _{k - 1}|\mathcal{M}_k)$ is inherently difficult to obtain. As $\mathit{\Gamma} _{k - 1}$ is a ternary-valued sequence, we draw $n$ particles $\{ \mathit{\Gamma} _{k-1}^i\} _{i = 1}^n$ from an importance distribution $q(\mathit{\Gamma} _{k - 1}|\mathcal{M}_k)$ to approximately compute $p(\mathit{\Gamma} _{k - 1}|\mathcal{M}_k)$, i.e.,
\begin{equation} \label{eq26}
p(\mathit{\Gamma} _{k - 1}|\mathcal{M}_k) \approx \sum\limits_{i = 1}^n {\omega _{k-1}^i} \delta (\mathit{\Gamma} _{k - 1} - \mathit{\Gamma} _{k - 1}^i)
\end{equation}
where $\delta (\cdot)$ is the standard Dirac delta function and the normalized particle weight $\omega _{k - 1}^i$ is associated with $\mathit{\Gamma} _{k - 1}^{i}$, which is given as
\begin{equation} \label{eq27}
\omega _{k-1}^i \propto \frac{{p(\mathit{\Gamma} _{k - 1}^i|\mathcal{M}_k)}}{{q(\mathit{\Gamma} _{k - 1}^i|\mathcal{M}_k)}}.
\end{equation}
Combining \eqref{eq25} with \eqref{eq26}, we obtain that
\begin{align*}
\widehat{\bm x}_{k|k}^{t} \approx  \sum\limits_{i = 1}^n {\omega _{k-1}^i} \mathbb{E}[\bm{x}_k^t|\mathit{\Gamma} _{k - 1}^i,\mathcal{M}_k].
\end{align*}
Similarly, the estimation error covariance matrix is given as
\begin{align*}
{\Sigma} _{k|k}^{t} \approx \sum\limits_{i = 1}^n {\omega _{k-1}^i} \mathbb{E}[(\bm{x}_k^t -\widehat{\bm x}_{k|k}^{t}){(\bm{x}_k^t - \widehat {\bm x}_{k|k}^{t})^T}|\mathit{\Gamma} _{k - 1}^i,\mathcal{M}_k].   
\end{align*}

Let $\widehat {\bm x}_{k|k}^{i}:= \mathbb{E}[\bm{x}_k^t|\mathit{\Gamma} _{k - 1}^i,\mathcal{M}_k]$ and  ${\Sigma} _{k|k}^{i}:=\mathbb{E}[(\bm{x}_k^t - \widehat {\bm x}_{k|k}^{i}){(\bm{x}_k^t - \widehat {\bm x}_{k|k}^{i})^T}|\mathit{\Gamma} _{k - 1}^i,\mathcal{M}_k]$. It should be noted that $\widehat {\bm x}_{k|k}^{i}$ and $\Sigma _{k|k}^{i}$ can be approximately computed by the EKF, i.e., 
\begin{align} \label{measuupdate}
& \widehat {\bm x}_{k|k}^{i}\approx\widehat {\bm x}_{k|k-1}^{i} + K_k^i(\bm{m}_k - {\bm h}(\widehat{\bm x}_{k|k-1}^{i})) \nonumber  \\
& {\Sigma} _{k|k}^{i} \approx {\Sigma} _{k|k - 1}^i - K_k^iH_k^{i}{\Sigma} _{k|k - 1}^i
\end{align}
where $K_k^i= {\Sigma} _{k|k - 1}^i(H_k^{i})^T(H_k^{i}{\Sigma} _{k|k - 1}^i(H_k^{i})^T+R)^{-1}$ and the Jacobian matrix evaluated at $\widehat {\bm x}_{k|k-1}^{i}$ is
\begin{equation}\label{jacbimatrix}
H_k^{i}=\frac{\partial}{\partial \bm{x}_k^t  } {\bm h}(\bm{x}_k^t)|_{\bm{x}_k^t =\widehat {\bm x}_{k|k-1}^{i}}.
\end{equation}

The remaining problem is how to recursively generate particles $\{\mathit{\Gamma} _{k-1}^i\}$ and compute their associated weights $\{\omega _{k-1}^i\}$.

\subsection{Importance Distribution}

If an importance distribution is chosen to factorize such that 
\begin{equation} \label{eq28}
q(\mathit{\Gamma} _{k - 1}|\mathcal{M}_k) = q({\gamma _{k - 1}}|\mathit{\Gamma} _{k - 2},\mathcal{M}_k)q(\mathit{\Gamma} _{k - 2}|\mathcal{M}_{k - 1}).
\end{equation}

Then, the particles $\mathit{\Gamma} _{k - 1}^i\sim q(\mathit{\Gamma} _{k - 1}|\mathcal{M}_k)$ can be obtained by augmenting each existing particle $\mathit{\Gamma} _{k - 2}^i\sim q(\mathit{\Gamma} _{k - 2}|\mathcal{M}_{k - 1})$ with the new state $\gamma _{k - 1}^i\sim q({\gamma _{k - 1}}|\mathit{\Gamma} _{k - 1},\mathcal{M}_k)$, recursively \cite{zhang2018bayesian}. To elaborate it, we express $p(\mathit{\Gamma} _{k - 1}|\mathcal{M}_k)$ in the following form
\begin{align*}
& p(\mathit{\Gamma} _{k - 1}|\mathcal{M}_k) \\
&\propto p({\bm{m}_k}|\mathit{\Gamma} _{k - 1},\mathcal{M}_{k - 1})p({\gamma _{k - 1}}|\mathit{\Gamma} _{k - 2},\mathcal{M}_{k - 1})p(\mathit{\Gamma} _{k - 2}|\mathcal{M}_{k - 1}).
\end{align*}

Jointly with \eqref{eq27}, it implies that
\begin{equation*}
\begin{split}
&\omega _{k-1}^i \\
&\propto \frac{{p({\bm{m}_k}|\mathit{\Gamma} _{k - 1}^i,\mathcal{M}_{k - 1})p(\gamma _{k - 1}^i|\mathit{\Gamma} _{k - 2}^i,\mathcal{M}_{k - 1})p(\mathit{\Gamma} _{k - 2}^i|\mathcal{M}_{k - 1})}}{{q({\gamma _k}|\Gamma _{k - 2}^i,\mathcal{M}_k)q(\mathit{\Gamma} _{k - 2}^i|\mathcal{M}_{k - 1})}}\\
&= \omega _{k - 2}^i\frac{{p({\bm{m}_k}|\mathit{\Gamma} _{k - 1}^i,\mathcal{M}_{k - 1})p(\gamma _{k - 1}^i|\mathit{\Gamma} _{k - 2}^i,\mathcal{M}_{k - 1})}}{{q(\gamma _{k - 1}^i|\mathit{\Gamma} _{k - 2}^i,\mathcal{M}_k)}}
\end{split}
\end{equation*}
where $p({\bm{m}_k}|\mathit{\Gamma} _{k - 1}^i,\mathcal{M}_{k - 1})$ is an approximately conditional Gaussian density, e.g.,
\begin{equation}\label{samdensity}
p({\bm{m}_k}|\mathit{\Gamma} _{k - 1}^i,\mathcal{M}_{k - 1}) \approx {\mathcal N}({\bm h}(\widehat{\bm x}_{k|k-1}^{i}),H_k^{i}\Sigma _{k|k - 1}^i(H_k^{i})^T + R).
\end{equation}

The degeneracy problem is essential to the success of particle sampling.  To alleviate it, there are two approaches of selecting the important distribution $q(\gamma _{k - 1}|\mathit{\Gamma} _{k - 2}^i,\mathcal{M}_k)$\cite{Arulampalam2002A}. One is $p(\gamma _{k - 1}|\mathit{\Gamma} _{k - 2}^i,\mathcal{M}_k)$, which minimizes a suitable measure of the degeneracy of the algorithm. The other is  $p(\gamma _{k - 1}|\gamma _{k - 2}^i)$, which makes it easy to draw particles and compute the importance weights. 

We  choose the later one as the importance distribution, e.g., $q(\gamma _{k - 1}|{\mathit{\Gamma}^i _{k - 2}},{\mathcal{M}_k}) = p(\gamma _{k - 1}|\gamma _{k - 2}^i)$  to simplify the process of drawing samples. Specifically, the new particle is  generated via the following distribution 
\begin{align} \label{eqdistribution}
\gamma _{k - 1}^i\sim p(\gamma _{k - 1}|\gamma _{k - 2}^i),
\end{align}
which is explicitly given in (\ref{eqtpmatrix})\footnote{If the transition probability matrix $P$ is unknown, it is suggested to directly sample $\gamma _{k - 1}^i$ via a uniform distribution.}.  Furthermore, the associated weights are updated as
\begin{align*}
\omega _{k - 1}^i \propto \frac{{p(\mathit{\Gamma} _{k - 1}^i|\mathcal{M}_k)}}{{q(\mathit{\Gamma} _{k - 1}^i|\mathcal{M}_k)}} = \omega _{k - 2}^ip({\bm{m}_k}|\mathit{\Gamma} _{k - 1}^i,\mathcal{M}_{k - 1}).
\end{align*}

Finally, the particle filter with a  resampling step is summarized in Algorithm \ref{rbpfalg}. 

\subsection{Jacobian Matrices of Sensor Models} \label{subJ}
The Jacobian matrix in (\ref{jacbimatrix}) is still pending, and depends on the sensor in use.  We show how to explicitly compute them in this subsection. 
\subsubsection{Camera Sensor} The noisy measurement ${\bm{m}^c_k}$  of the camera is given as
	\begin{equation} \label{eqmea}
		{\bm{m}^c_k} = {\bm h}^c({\bm c}_k^{ta}) + {\bm{v}^c_k}
	\end{equation}
	where $\{\bm{v}^c_k\}$ is a white Gaussian noise, i.e.,  ${\bm{v}^c_k}\sim{\mathcal  N}(0,R^c)$. ${\bm c}_k^{ta}$ is the relative position of the GMT to the camera in the camera frame (c.f. Fig.~\ref{figpos})  and ${\bm h}^c({\bm c}_k^{ta})$ returns the coordinates of ${\bm c}_k^{ta}$ on the image plane, which is the projection of the GMT onto the image plane and is defined in  (\ref{eqcam}). 
	
	By (\ref{jcbcam}) and (\ref{jacbimatrix}), the Jacobian matrix of ${\bm h}^c({\bm c}_k^{ta})$ with respect to ${\bm x}_{k}^t$  is given as
	
\begin{equation}\label{jacbmac}
H_k^c:=\frac{{\bm h}^c({\bm c}_k^{ta})}{\partial {\bm x}_{k}^t}=J^c({\bm c}_k^{ta})\frac{\partial{\bm c}_k^{ta}}{\partial {\bm x}_k^t}.
\end{equation}

Next, we show how to compute ${\partial{\bm c}_k^{ta}}/{\partial {\bm x}_k^t}$.  When using a camera to take measurements, we adopt a gimbal platform to keep the GMT in the FOV of the camera. The platform is composed by a yaw gimbal and a pitch gimbal. The pitch gimbal can only rotate around the pitch axis with a pitch angle ${\theta ^c}$, and the yaw gimbal can only rotate around the yaw axis with a yaw angle ${\psi ^c}$. Consequently, the transformation matrix from the body frame of the UAV to the camera frame is computed as
\begin{equation*}
	\begin{split}
		{C^{ca}_k} 
		= \bmatri
				{\cos {\theta ^c_k}}&0&{\sin {\theta ^c_k}}     \\
				0&1&0\\
				{ - \sin {\theta ^c_k}}&0&{\cos {\theta ^c_k}}
			\ematri\bmatri
			{\cos {\psi ^c_k}}&{\sin {\psi ^c_k}}&0\\
			{ - \sin {\psi ^c_k}}&{\cos {\psi ^c_k}}&0\\
			0&0&1
		\ematri
	\end{split}
\end{equation*}
where ${\theta ^c_k} \in [0,\pi /2]$ is the pitch angle, and ${\psi ^c_k} \in [-\pi ,\pi)$ is the yaw angle.

Since the UAV is flying with a constant altitude, the transformation matrix from the inertial frame to the body frame is only relative to the heading angle of the UAV. Thus, it is given as
\begin{equation}
	C^{ai}_k = \bmatri
			{\cos \psi^a _k}&{\sin \psi^a _k}&0\\
			{ - \sin \psi^a _k}&{\cos \psi^a _k}&0\\
			0&0&1
		\ematri
	\end{equation}
	where $\psi^a_k  \in [ - \pi ,\pi)$ is the heading angle of the UAV.
	
	Overall, the transformation matrix from the inertial frame to the camera frame is
	\begin{equation}
		C^{ci}_k = C^{ca}_k C^{ai}_k.
	\end{equation}

	Consider Fig.~\ref{figpos}, where $\{I\}$ is the inertial frame, and in the inertial frame  denote the positions of the UAV\footnote{The positions of the UAV, the camera and the radar are assumed to be of the same.} and the GMT respectively by ${\bm i}^a_k=[x^a_k,y^a_k,z^a_k]^T $ and ${{\bm i}^t_k} =[x^t_k,y^t_k,z^t_k]^T$,  which in implementation is replaced by its one-step prediction as in \eqref{jacbimatrix}. Then, ${\bm i}_k^{ta}:={\bm i}^t_k-{\bm i}^a_k$ is the relative position of the GMT to the camera in the inertial frame. 
	
	By the definition of ${\bm c}_k^{ta}$, it follows that ${\bm i}_k^{ta}=C_k^{ic}{\bm c}_k^{ta}$. Since $C_k^{ic}=(C_k^{ci})^{-1}$, it implies that ${\bm c}_k^{ta}=C_k^{ci}{\bm i}_k^{ta}$ and
	$$
	\frac{\partial{\bm c}_k^{ta}}{\partial {\bm x}_k^t}=[C_k^{ci}~~0_{3\times 2}], 
	$$
	where $0_{3\times 2}$ denotes a zero matrix with a compatible dimension. Jointly with (\ref{jcbcam}) and (\ref{jacbmac}), the Jacobian matrix of the camera model can be computed. Note that ${\bm i}^a_k$  is directly obtained by using the POS in the UAV. 

\subsubsection{Radar Sensor}
		The noisy measurement  of a radar is given by
		\begin{equation} \label{eqmeara}
			{\bm{m}^{r}_k} = h^{r}({\bm i}_k^{ta}) + {\bm{v}^{r}_k}
		\end{equation}
		where $\{\bm{v}^{r}_k\}$ is a white Gaussian noise, i.e., ${\bm{v}^{r}_k}\sim{\mathcal  N}(0,R^{r})$, $h^{r}({\bm{x}_k})$ is defined in \eqref{eqrada} and ${\bm i}_k^{ta}$ is the relative position of the GMT to the radar in the inertial frame. 

Then, the Jacobian matrix of $h^{r}({\bm i}_k^{ta})$ with respect to ${\bm x}_{k}^t$  is easily given as
	
\begin{equation}\label{jacbmac1}
H_k^r:=J^r({\bm i}_k^{ta})\frac{\partial {\bm i}_k^{ta}}{\partial {\bm x}_k^t}=[J^r({\bm i}_k^{ta})~~0_{2\times 2}]
\end{equation}
where $J^r(\cdot)$ is defined in (\ref{jcbradar})  and ${\bm i}_k^{ta}$ is computed as in the case of camera sensor.

\begin{algorithm} [!t]
	\caption{RBPF for the GMT with unknown maneuver}
    \label{rbpfalg}
	\begin{enumerate}\renewcommand{\labelenumi}{\rm(\alph{enumi})}
		\item \textbf{Initialization:} For  $i = 1,\ldots,n$, draw $n$ particles $\gamma _{0}^i$ from the prior $p(\gamma _{0})$ and let $\widehat{\bm x}_{0|0}^{i} = {\bm{x}_0}$, ${\Sigma} _{0|0}^{i} = {{\Sigma} _0}$, $\omega _{0}^i = 1 / n$.
		\item \textbf{Updates:} For $k\ge1$ and $i=1,\ldots,n$, the UAV does the following updates:
		\begin{itemize}
		\item \textbf{Time update:} The state prediction and its covariance matrix are updated by
			           \begin{align*}
			          \widehat {\bm x}_{k|k - 1}^i &= {F_{k-1}}\widehat {\bm x}_{k - 1|k - 1}^i + B_{k-1}\bm{u}^t({\gamma _{k - 1}^i}),\\
			           {\Sigma} _{k|k - 1}^i &= {F_{k-1}}{\Sigma} _{k - 1|k - 1}^i{F_{k-1}^T} + G_{k-1}QG_{k-1}^T .
			           \end{align*}
			\item \textbf{Measurement update:} The UAV receives a measurement $\bm{m}_k$ and does the following updates
			\begin{align*}
			& \widehat {\bm x}_{k|k}^{i}=\widehat{\bm x}_{k|k-1}^{i} + K_k^i(\bm{m}_k - {\bm h}(\widehat{\bm x}_{k|k-1}^{i}))   \\
			& {\Sigma} _{k|k}^{i} = {\Sigma} _{k|k - 1}^i - K_k^iH_k^{i}{\Sigma} _{k|k - 1}^i
			\end{align*}
			where $K_k^i$ and $H_k^{i}$ are given in (\ref{measuupdate}).
		\end{itemize}
		\item  \textbf{Importance weight update, and resampling:} If $k\ge 2$, then
		\begin{itemize}
		\item The normalized importance weights are updated as
			\begin{align*}
			\omega _{k - 1}^i  \propto  \omega _{k - 2}^ip({\bm{m}_k}|\mathit{\Gamma} _{k - 1}^i,\mathcal{M}_{k - 1})
			\end{align*} 
				where $p({\bm{m}_k}|\mathit{\Gamma} _{k - 1}^i,\mathcal{M}_{k - 1})$ is given in (\ref{samdensity}).
		\item  Compute an estimate of the effective number of particles  
			\begin{align*}
			n_{\rm{eff}} = \frac{1}{{\sum\nolimits_{i = 1}^n {{{(\omega _{k-1}^i)}^2}} }} .
			\end{align*}
			\item Given a resampling threshold $n^{\rm{c}}>0$. If  $n^{\rm{eff}} < n^{\rm{c}}$, then perform resampling.  Take $n$ new samples $\gamma _{k - 1}^{i*}$ with replacement from the set $\{ \gamma _{k - 1}^i\}_{i=1}^n $ according to the probability distribution that  
			\begin{align*}
			\Pr\{\gamma _{k - 1}^{i*} = \gamma _{k - 1}^i\} = \omega _{k-1}^i.
			\end{align*}
		    \item  Let $(\gamma _{k - 1}^i,\widehat{\bm x}_{k|k}^i,{\Sigma} _{k|k}^i) = (\gamma _{k - 1}^{i*},\widehat{\bm x}_{k|k}^{i*}, {\Sigma} _{k|k}^{i*})$ and $\omega _{k-1}^i=1/n$. 
			\end{itemize}
		\item  \textbf{Output:} The state estimate of the GMT and its covariance matrix are given by
			 \begin{align*}
			 \widehat{\bm x}_{k|k}^{t} &=  \sum\limits_{i = 1}^n {\omega _{k-1}^i} \widehat {\bm x}_{k|k}^{i} \\
			 \Sigma_{k|k}^{t} &=  \sum\limits_{i = 1}^n {\omega _{k-1}^i} \Sigma_{k|k}^{i}.
			 \end{align*}
		\item \textbf{Sampling:} Let $k \leftarrow k+1$ and for $i=1, \ldots, n$, draw $\gamma_{k-1}^{i}\sim p(\gamma _{k - 1}|\gamma _{k - 2}^i)$.
		\end{enumerate}		
\end{algorithm}

\subsection{Certainty Equivalence for the Flight Controller}
To obtain the flight controller for the UAV, we adopt the principle of certainty equivalence \cite{bertsekas1995dynamic} by directly replacing the exact states of the GMT $[x^t_k, y^t_k, v_k^{tx}, v_k^{ty}]^T$ with their estimates $[\widehat x^t_k, \widehat y^t_k, \widehat v_k^{tx},\widehat v_k^{ty}]^T$, which is obtained by running Algorithm \ref{rbpfalg}. Specifically, let $[\widehat{v}_k^{d}, \widehat{\psi}_k^{d}]^T$ be the estimated desired velocity and heading angle of the UAV, which are obtained by using the state estimates of the GMT in (\ref{equmov}).  Then, we obtain that
 \begin{align*}
v^d_k &= \widehat{v}_k^{d} + \widetilde v_k^{d}, \\
{\psi ^d_k} &= \widehat{\psi}_k^{d} + \widetilde \psi_k^{d},
\end{align*}
where $ \widetilde v_k^{d}$ and $ \widetilde \psi_k^{d}$ denote  the estimation error of the desired velocity  $v_k^d$  and heading angle $\psi_k^d$ of the UAV. Usually, it is impossible to compute the exact error bounds of a nonlinear filter, which renders that we are unable to rigorously prove the stability of overall systems with the estimated motion state.  However, the simulation results  indicate the proposed controller in this paper indeed is able to complete the loitering task. 

\section{Simulation}\label{sec_simulation}
In this section, we perform simulations to illustrate the effectiveness of the proposed flight controller. 

\subsection{Simulation Setup}
The GMT starts at $[0,100,0]^T$\si{m} with an  initial speed $8$\si{m/s}  and an initial heading angle $\pi /4$\si{rad}. The initial control input is set to be zero, i.e., $\gamma _0 = 1$. Moreover, the three switching commands are given as $\bm{u}^t(1) = {[0,0]^T}$\si{m/s^2}, $\bm{u}^t(2) = {[ - 1,1]^T}$\si{m/s^2}, and  $\bm{u}^t(3) = {[1, - 1]^T}$\si{m/s^2}. The transition probability matrix of the three modes is given in \eqref{eqtpmatrix}. 
The process noise $\bm{w}^t$ of the GMT is a white Gaussian noise with covariance $Q^t = \rm{diag}(0.3^2,0.3^2,0.1^2)$. 

The UAV starts at the position $[- 300, 100, 50]^T$\si{m}, and keeps its altitude invariant. The initial speed and initial heading angle are $10 $\si{m/s} and $ - \pi /2$\si{rad}, respectively. Moreover, the turning rate is restricted to the interval  $[- 0.2,0.2] $\si{rad/s}.  The random disturbances $\{\bm{w}_k^{a}\}$ to the UAV are given as $w_k^{av}\sim \mathcal{N}(0,0.1^2)$ and $w_k^{a\psi}\sim \mathcal{N}(0,0.02^2)$.

\subsection{Comparisons of Control Methods under Exact States of the GMT}

\begin{figure}[!t]
	\centerline{\includegraphics[width=0.8\linewidth]{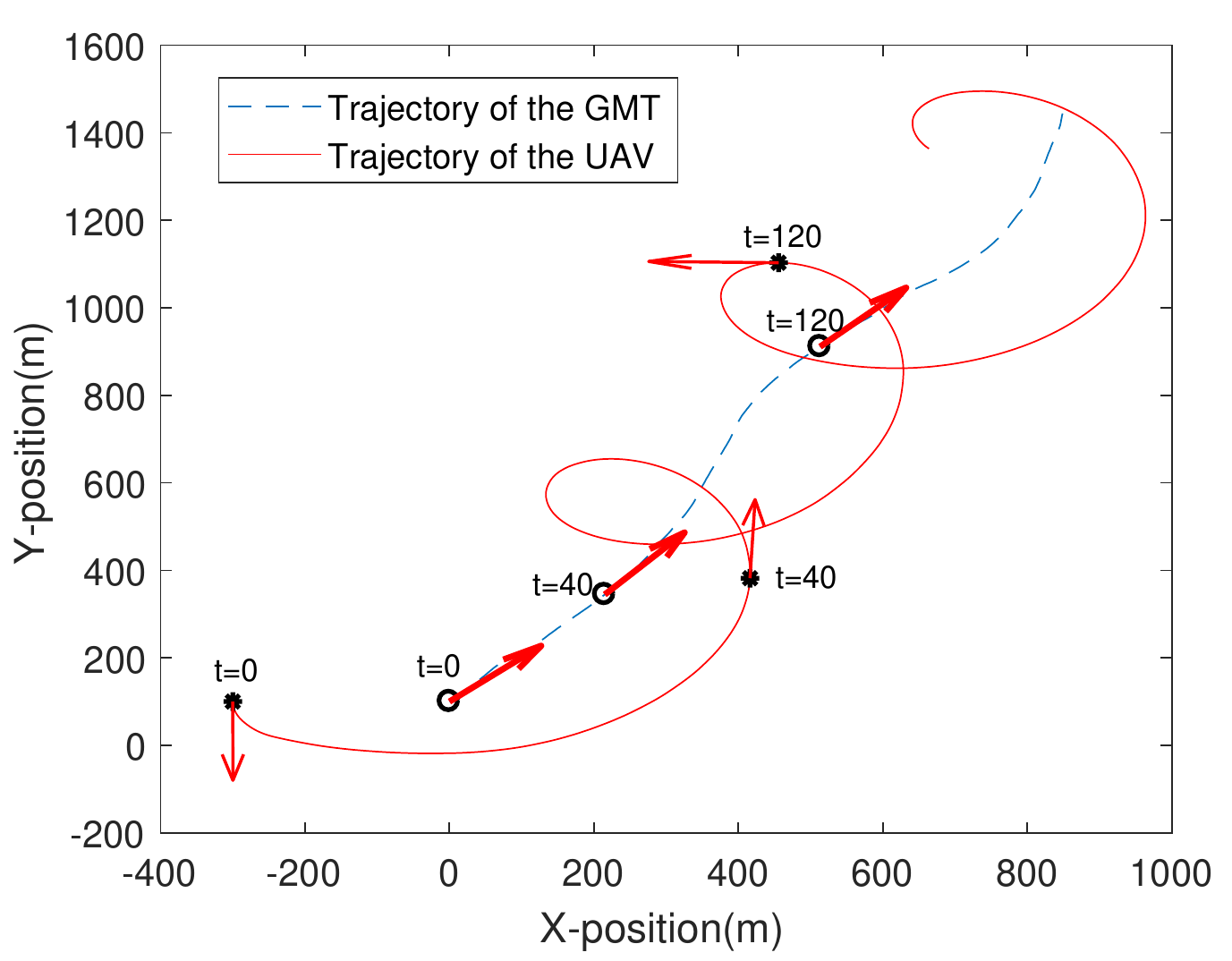}}
	\caption{Trajectories of the GMT and the UAV in the X-Y plane.}
	\label{figpath3}
\end{figure}

\begin{figure}[!t]
	\centering		
	\includegraphics[width=0.8\linewidth]{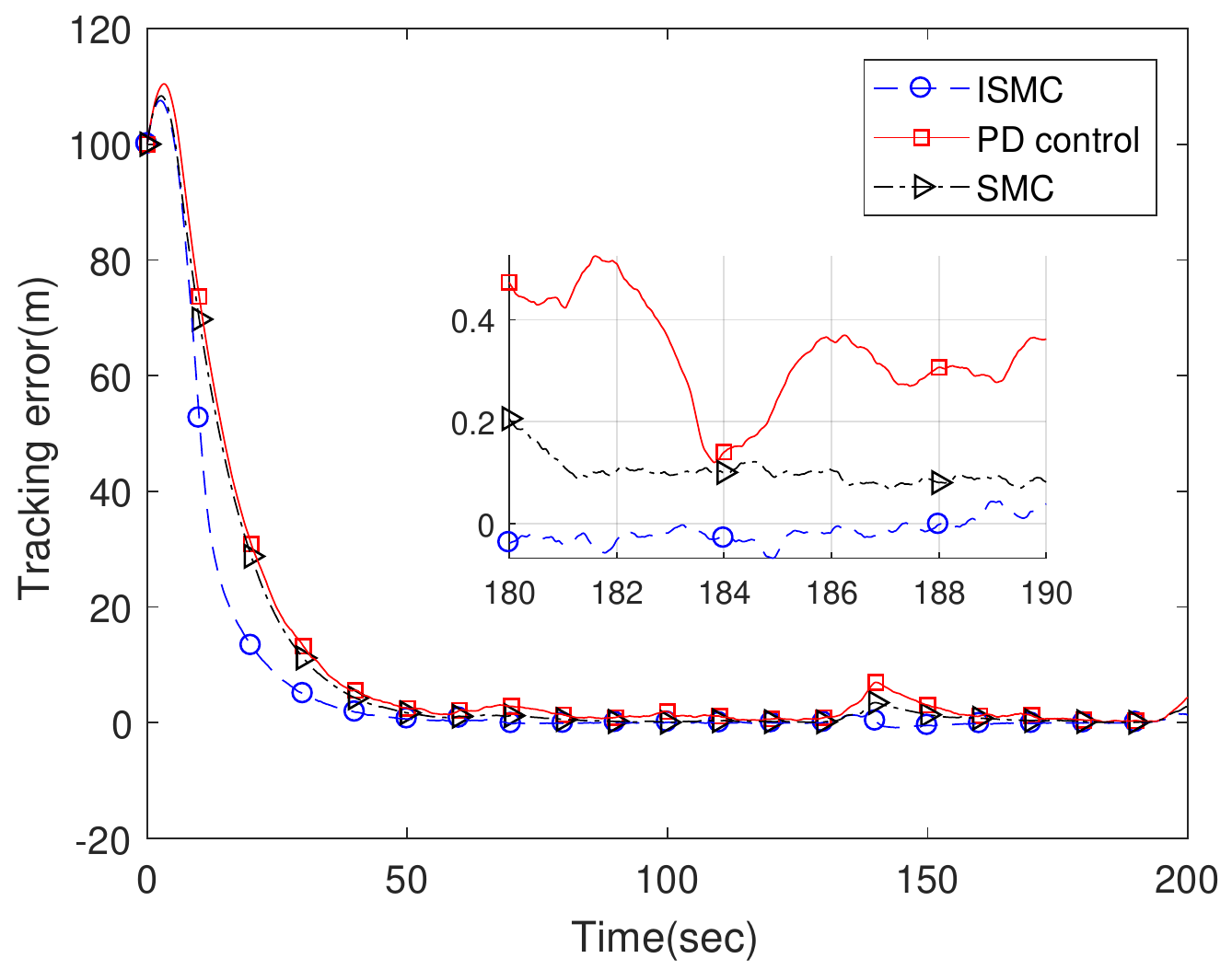}	
	\caption{Tracking errors under different control methods.}
	\label{fig1}
\end{figure}

In this subsection, we assume that the state of the GMT is known to the UAV. Since the  guidance vectors in \eqref{eqlya} can only be given {\em online}, i.e., guidance vectors after time $k$ are unavailable to the design of the $k$-th time input of the UAV,  more advanced controllers do not always work, e.g., the model predictive control cannot be applied here as it relies on {\em future}  guidance vectors \cite{Camacho2004Model}. Thus, we only compare the proposed ISMC with the PID control \cite{Dobrokhodov2008Vision,Frew2008Coordinated,Zhang2010Vision}, and the standard SMC \cite{Perruquetti2002Sliding} for the UAV with dynamics in (\ref{equ4}). The trajectories of the UAV and GMT are shown in Fig.~\ref{figpath3},  where circles and stars represent their positions at different time instants, and  arrows denote course directions. The desired distance from the UAV to the GMT is set to $200$\si{m}. 

The comparison is depicted in Fig.~\ref{fig1}, which shows that the proposed ISMC has the shortest setting time with zero steady-state error. We shall further test its effectiveness by using the estimated states via the RBPF. 

\subsection{Motion Estimation} \label{seccam}
In this section, we only test the estimation performance of the proposed RBPF. If the transition probability matrix $P$ in \eqref{eqtpmatrix} for the maneuver process is unknown, each element of $P$ is set to $1/3$. Moreover, we adopt the Monte Carlo method by independently repeating  $1000$ experiments  to compute the root-mean-square error (RMSE) of the state estimate of the GMT, i.e., 
\bee\label{rmsesim}
\text{RMSE}_k:=\left(\bE[(\widehat{x}_k^t-{x}_k^t)^2+(\widehat{y}_k^t-{y}_k^t)^2]\right)^{1/2},
\ene
where $[\widehat{x}_k^t,\widehat{y}_k^t]^T$ is produced by the RBPF.  All simulations are of the same start position with $100$ particles in the RBPF. 

For comparison, an EKF is also designed for the motion estimation.  Note that the EKF requires the input to the GMT, which is unfortunately unknown in our setting. To solve it, we observe that the stationary distribution of the Markov chain under the transition probability matrix $P$ in \eqref{eqtpmatrix} is a uniform distribution. Thus, we  randomly sample a control input from the set $\{\bm{u}^t(1),\bm{u}^t(2),\bm{u}^t(3)\}$ with equal probability for the EKF.

\subsubsection{Camera Sensor}

\begin{table}[t!] \label{tabsliding}%
    \caption{Parameters of the ISMC}
    \centering	
	\begin{tabular}{|c|c|c|c|}	
		\hline
		\diagbox{Superscript}{Parameter}  &~~$W$~~ &~~$M$~~ &~~$C$~~\\
		\hline
		v                                                       & 0.2         & 5.0          & 5.0  \\
		$\psi$                                                 & 0.04      & 0.6       & 3.0  \\
		\hline
	\end{tabular}%
\end{table}%

The sampling frequency of the camera sensor is $25$\si{Hz} and
the measurement noise is Gaussian white noise with covariance $R^c = \rm{diag} (0.03^2,0.03^2)$. Note that these parameters in real experiments of \cite{Dobrokhodov2008Vision} are set as $30$\si{Hz} and $\rm{diag} (0.02^2,0.02^2)$, respectively. Fig.~\ref{fig4}  illustrates the RMSE of the RBPF and EKF. One can observe that the performance of the RBPF is better than that of the EKF and is not significantly degraded even if the transition probability matrix $P$ is unknown, and both cases return favorable estimation performance. 

%

\begin{figure}[!t]
	\centering		
	\includegraphics[width=0.8\linewidth]{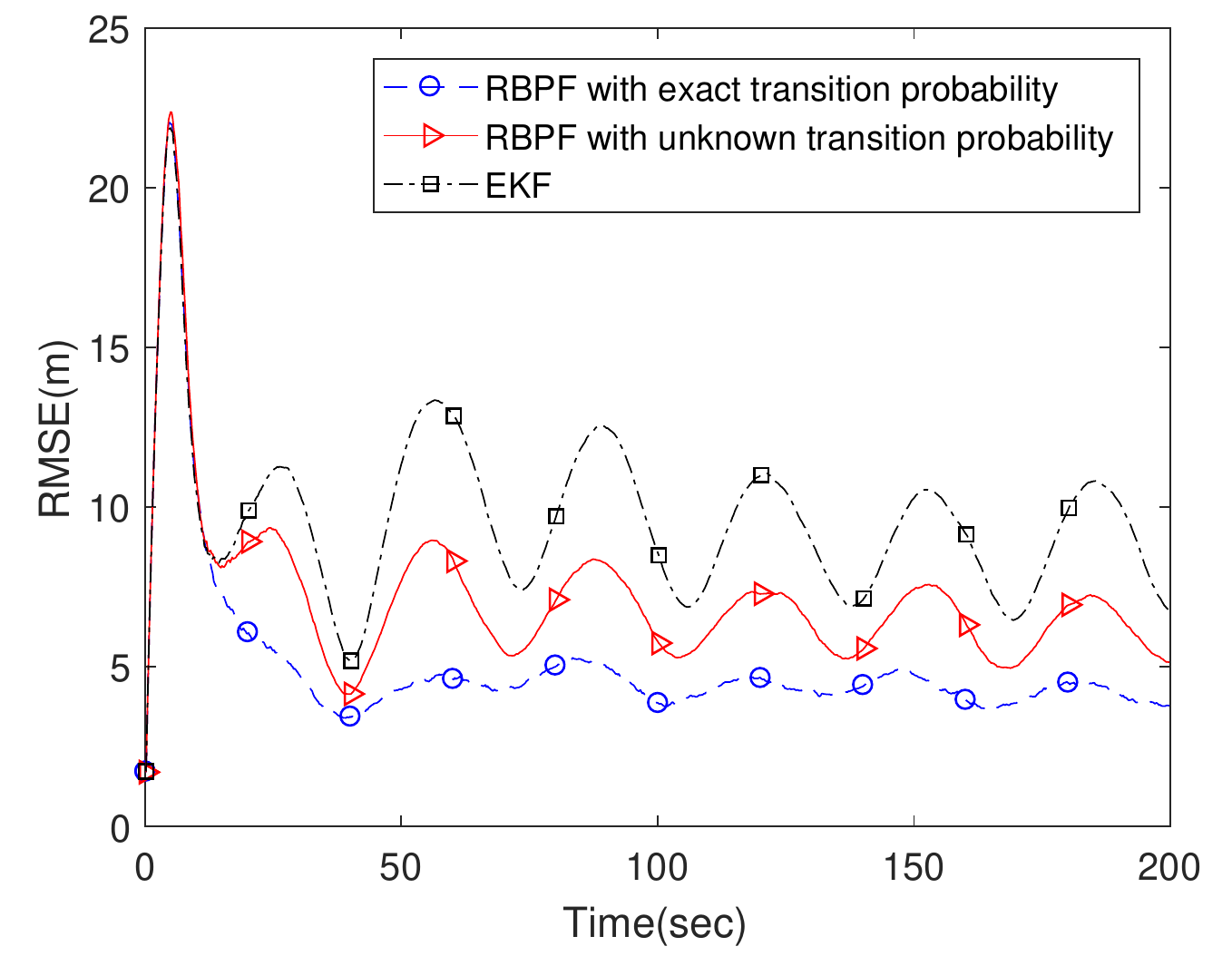}	
	\caption{RMSE of the RBPF and EKF with a camera sensor.}
	\label{fig4}
\end{figure}

\subsubsection{Radar Sensor}
For a radar sensor, we follow the same setting in \cite{Averbuch1991Radar} where the sampling frequency is $10$\si{Hz} and the covariance is $R^r =\rm{diag}(2.0^2,0.01^2)$. From Fig. \ref{figradardis1}, the same conclusion can be made as in the case of the camera sensor. 

\begin{figure}[!t]
	\centering		
	\includegraphics[width=0.8\linewidth]{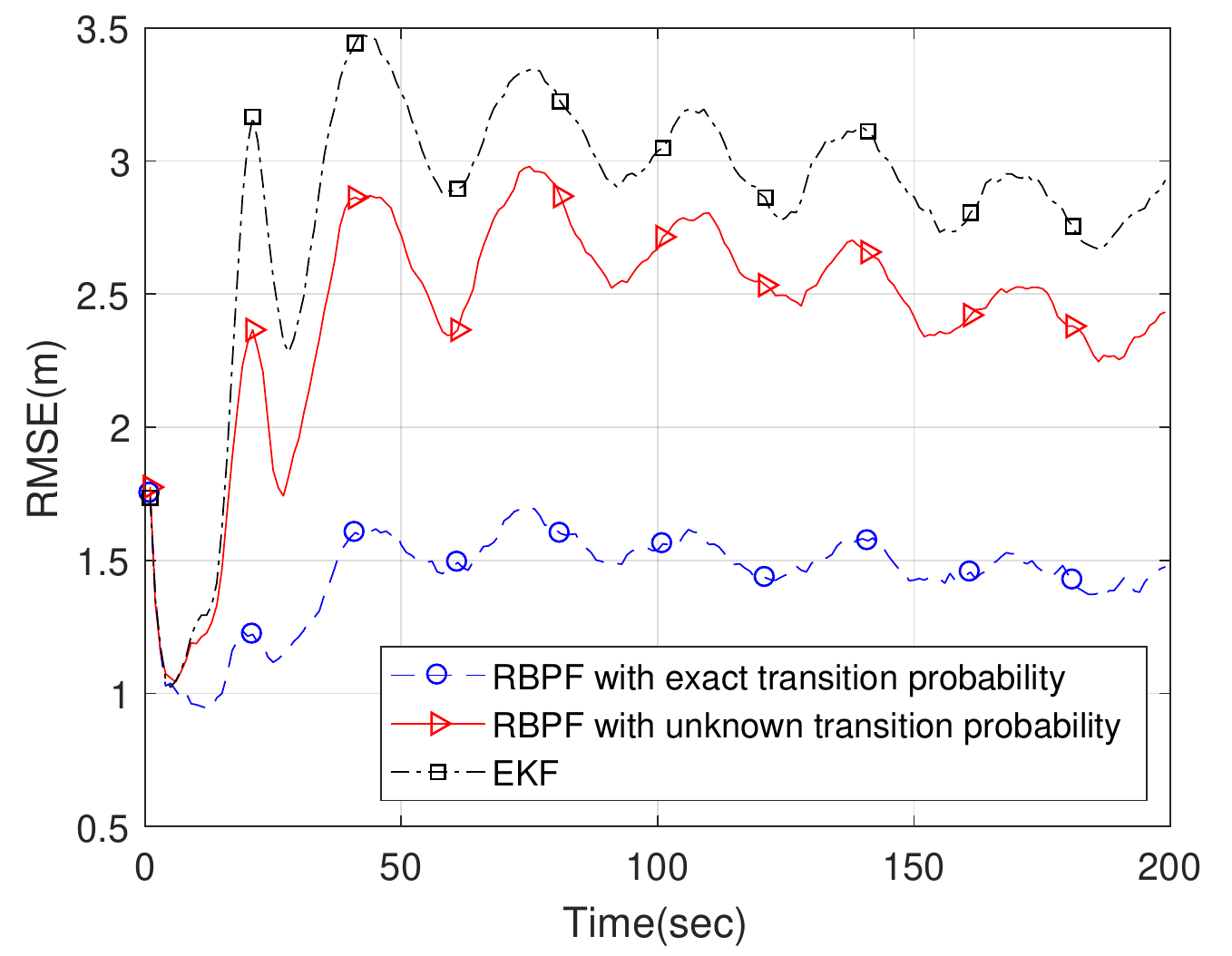}	
	\caption{RMSE of the RBPF and EKF with a radar sensor.}
	\label{figradardis1}
\end{figure}
Thus, both cases consistently verify the effectiveness of the RBPF.

\subsection{Loitering Performance with Estimated States of the GMT}

\begin{figure}[!t]
	\centering		
	\includegraphics[width=0.8\linewidth]{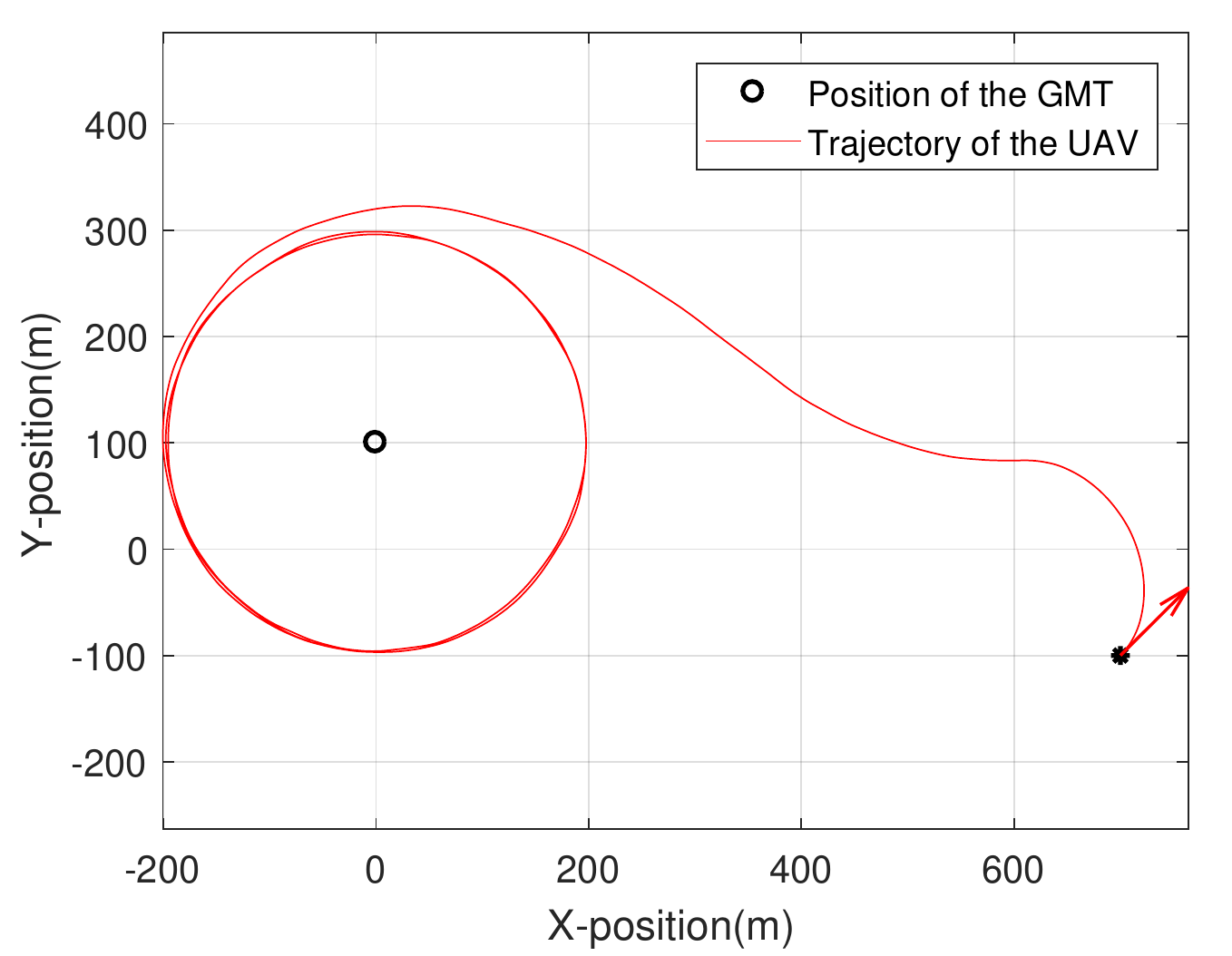}	
	\caption{Trajectory of the UAV when the GMT is stationary.}
	\label{fig5}
\end{figure}
\begin{figure}[!t]
	\centering		
	\includegraphics[width=0.8\linewidth]{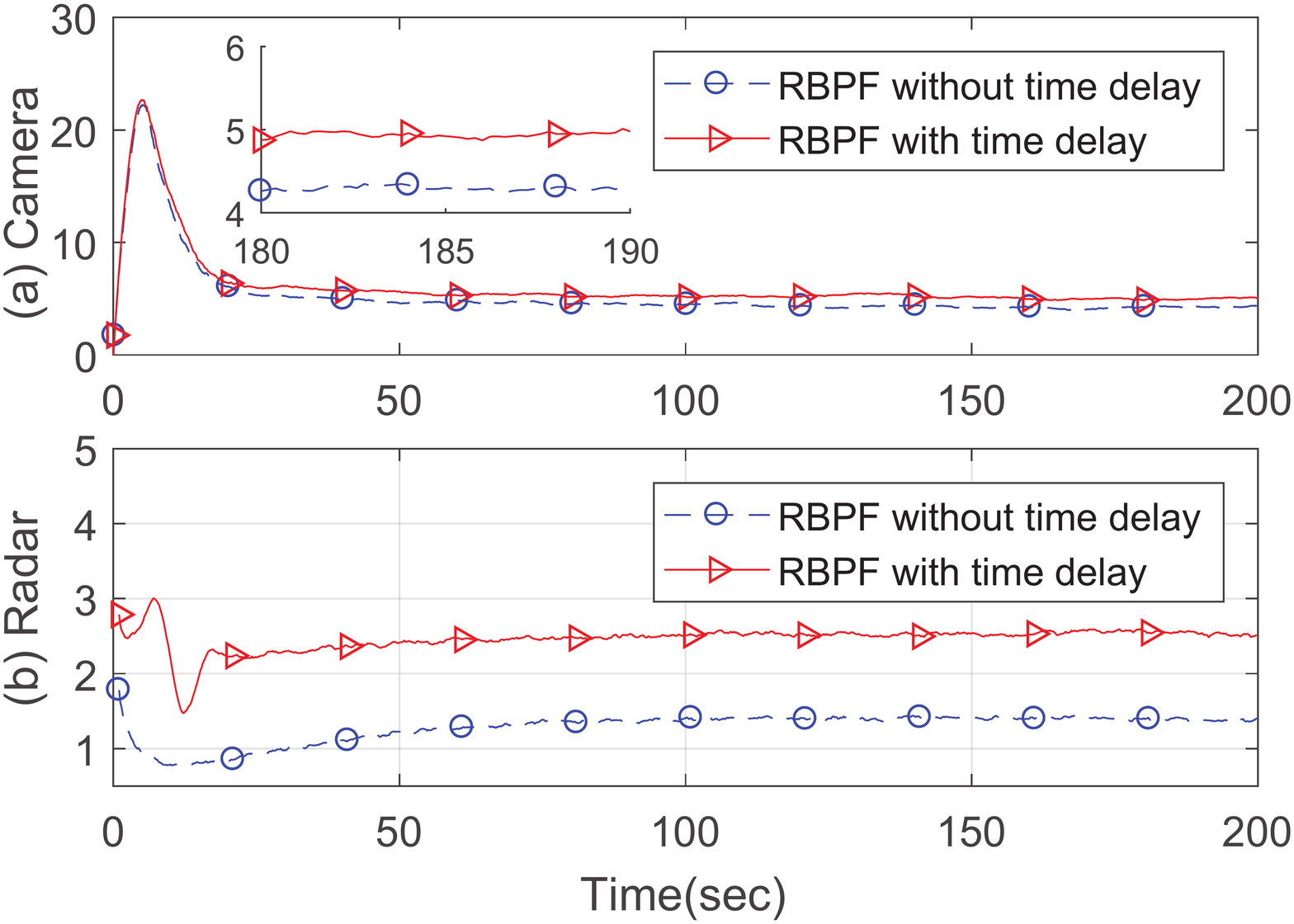}	
	\caption{RMSE of the RBPF with: (a) a camera sensor; (b) a radar sensor.}
	\label{fig3}
\end{figure}
In this subsection, we test the loitering performance of the UAV with estimated states of the GMT in \eqref{gmtdyn}. The sampling frequencies are of the same as that in Section \ref{seccam} with the consideration of  $0.1$\si{s} time delay in sensor measurements. Note that the number of particles is still $100$.  

Firstly,  we consider the dynamics of the UAV in \eqref{equ4}. When the GMT is stationary, Fig.~\ref{fig5} depicts the trajectory of the UAV under the proposed controller. Clearly, the UAV finally circumnavigates the GMT with a desired radius $200$\si{m}. 

Then, $\{\gamma_k\}$ is modeled as a nine state Markov chain, which corresponds to nine modes of maneuver: $\bm u^t(1)=[0,0]^T$, $\bm u^t(2)=[1,0]^T$, $\bm u^t(3)=[1,1]^T$, $\bm u^t(4)=[0,1]^T$, $\bm u^t(5)=[-1,1]^T$, $\bm u^t(6)=[-1,0]^T$, $\bm u^t(7)=[-1,-1]^T$, $\bm u^t(8)=[0,-1]^T$, $\bm u^t(9)=[1,-1]^T$. Moreover, all the diagonal elements of the transition probability matrix $P$ are $0.6$ and other elements are set to $0.05$.  For the implementation of the RBPF, $P$ is assumed to be unknown and each element is simply set as $1/9$. Fig.~\ref{fig3} shows the RMSE defined in (\ref{rmsesim}) over $1000$ times under different sensors. 

\begin{figure}[!t]
	\centering		
	\includegraphics[width=0.8\linewidth]{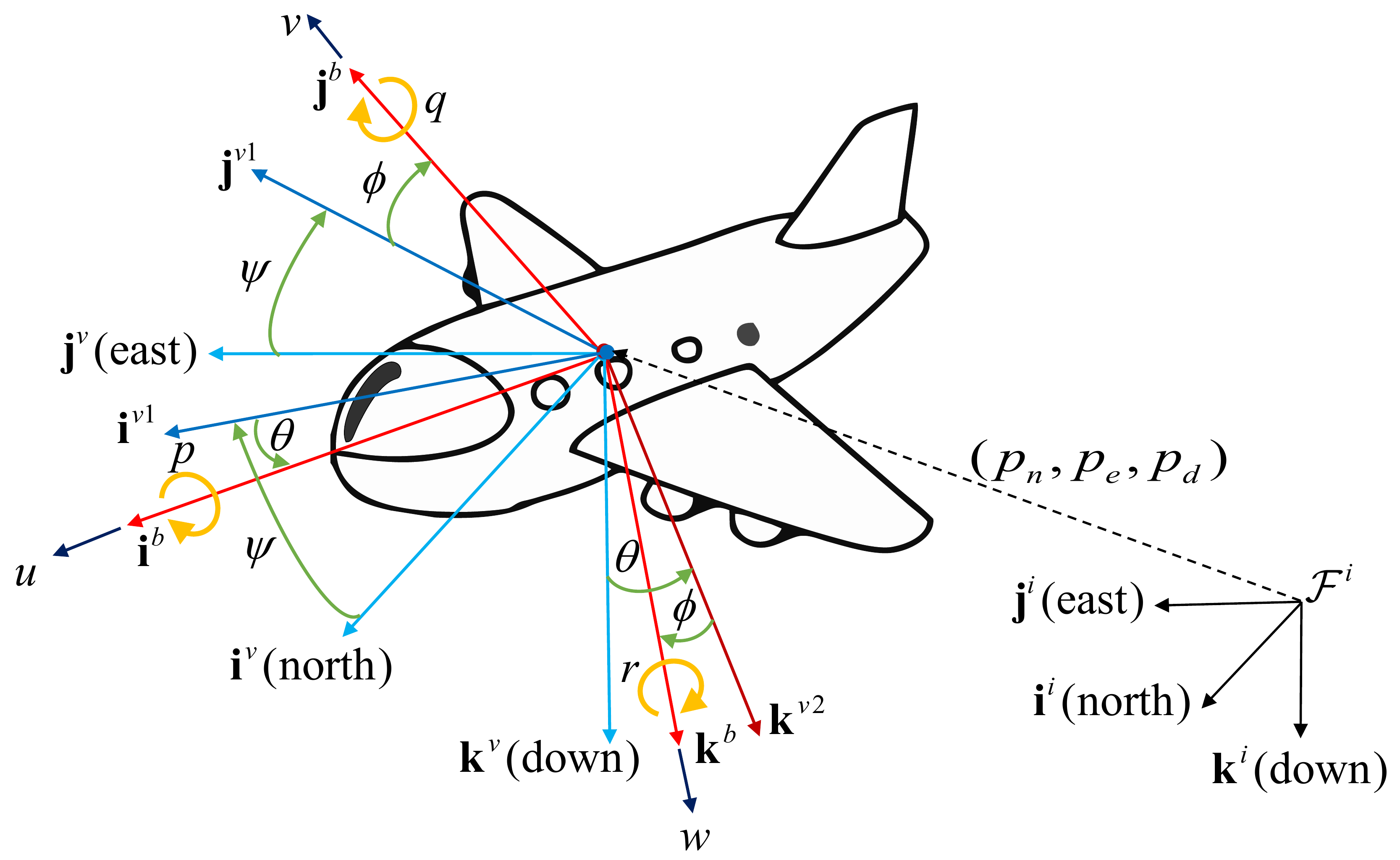}	
	\caption{6-DOF fixed-wing UAV and its associated coordinate frames \cite{Beard2012Small}.}
	\label{fig6}
\end{figure}

Then, a 6-DOF fixed-wing UAV \cite{Beard2012Small} is adopted to test the effectiveness of the proposed flight controller,  see Fig.~\ref{fig6} where $[p_n,p_e,p_d]^T$ and $[\phi,\theta,\psi]^T$ are the position
and orientation of the UAV in the inertial coordinate frame, respectively. 
$[u,v,w]^T$ and $[p,q,r]^T$ are linear velocities and angular rates in the body frame. Due to page limitation, we omit details of the mathematical
model of the UAV, which can be found in \cite{Beard2012Small}, and adopt codes from \cite{small} for the model. To complete the loitering task with a camera sensor, we design the controller by the proposed method of this work. The 3D trajectories of the GMT and the UAV are shown in Fig.~\ref{fig7}, where the altitude is controlled to $100$\si{m} under the altitude controller in \cite{small}. Moreover, the trajectories on X-Y plane, the trajectory in Z direction, the estimation error, and tracking error of one run are all depicted in Fig.~\ref{fig8}. Simulation results in this case also indicate that the proposed flight controller is effective. 

\begin{figure}[!t]
	\centering		
	\includegraphics[width=0.8\linewidth]{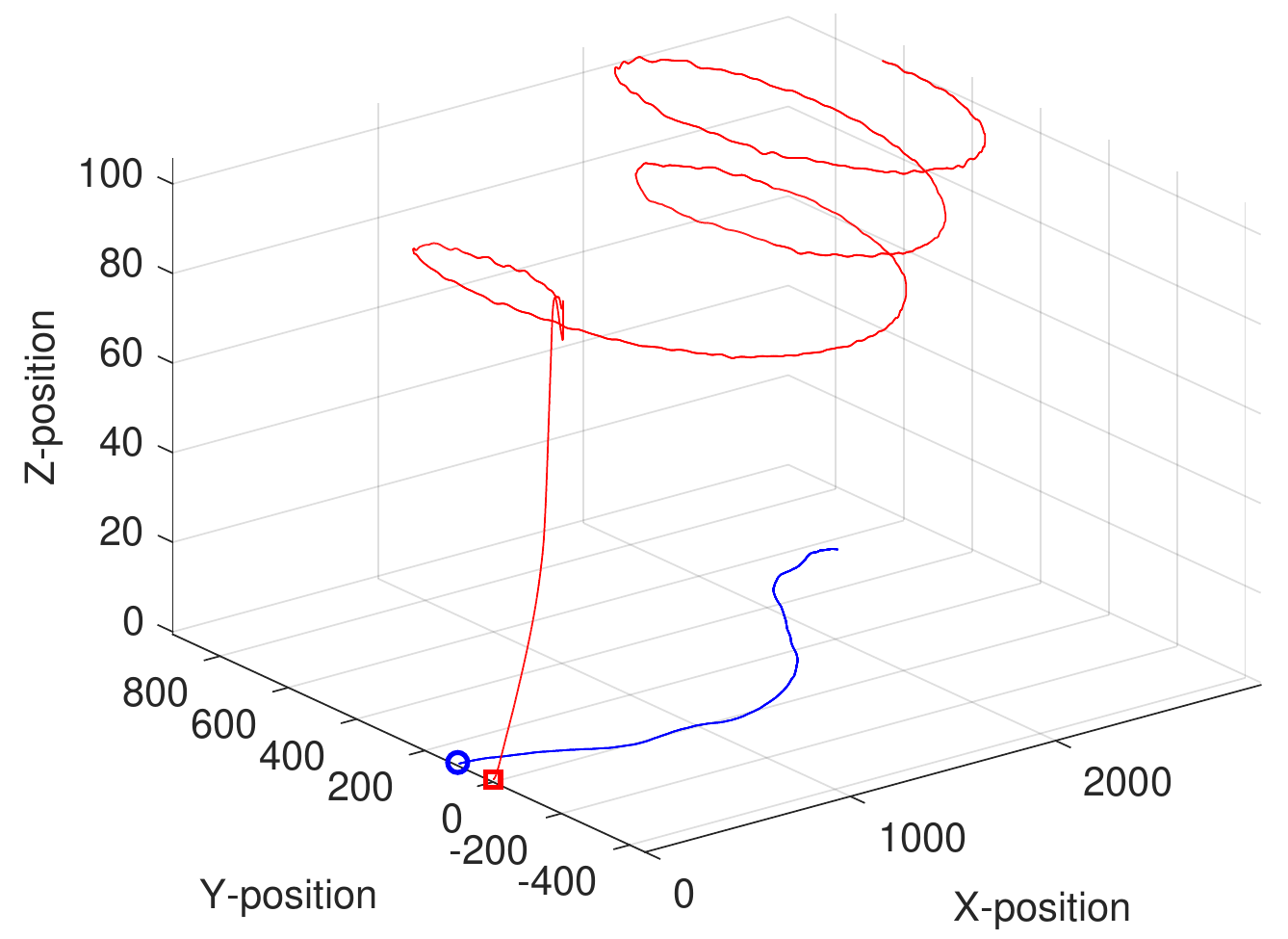}	
	\caption{3D trajectories of the GMT and the UAV.}
	\label{fig7}
\end{figure}

\begin{figure}[!t]
	\centering		
	\includegraphics[width=0.7\linewidth]{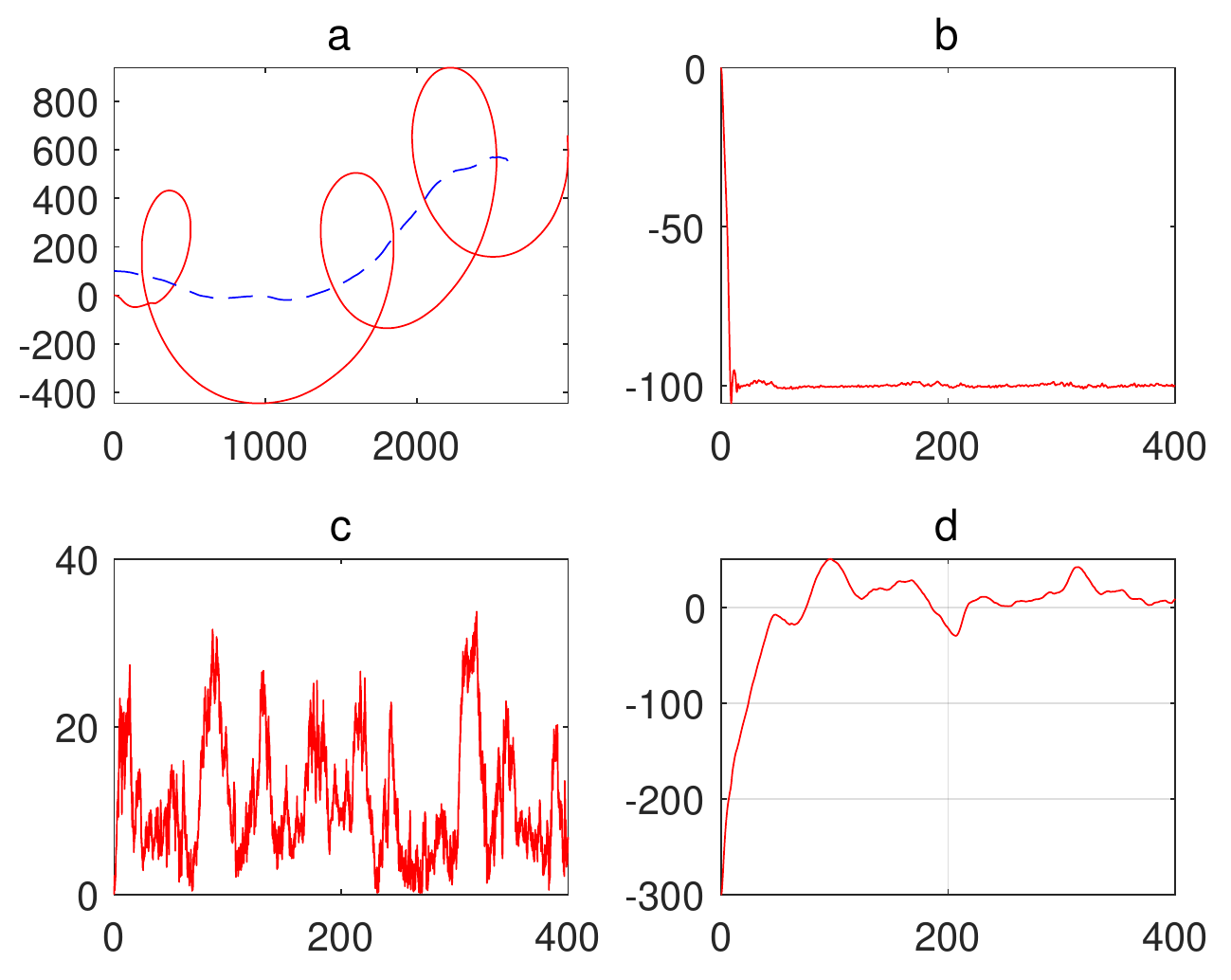}	
	\caption{Simulation results: (a) trajectories of the GMT and the UAV projected in XY plane; (b) the altitude trajectory of the UAV; (c) the absolute estimation error of the GMT in one run; (d) the tracking error  in one run.}
	\label{fig8}
\end{figure}

\section{Conclusion}
\label{sec_conclusion}
In this paper, we proposed a discrete-time ISMC for the UAV to loiter over a  GMT with three possible maneuvering states. To achieve it, a discrete-time guidance vector field was designed by assuming that the motion state of the GMT is known. Then, we designed a RBPF to simultaneously estimate the maneuver and the motion state of the GMT by using the measurements of a vision camera or a radar. Simulations finally validated our theoretical results. 
\bibliographystyle{IEEEtran}
\bibliography{bib/mybib}

\end{document}